\documentclass{elsarticle}

\usepackage{amssymb, amsmath}
\usepackage{stmaryrd}

\usepackage{multirow,tcolorbox}
\usepackage{url}
\usepackage{hyperref} 
\usepackage{natbib}

\usepackage{tikz}

\usetikzlibrary{calc,trees,positioning,arrows,chains,automata,shapes.geometric,%
	decorations.pathreplacing,decorations.pathmorphing,shapes,%
	matrix,shapes.symbols}

\usetikzlibrary{shapes,arrows,automata,positioning,calc}
\usetikzlibrary{fit,backgrounds}
\usetikzlibrary{decorations.pathreplacing}
\tikzset{
	>=stealth',
	punktchain/.style={
		rectangle,
		rounded corners,
		draw=black, very thick,
		text width=6em,
		minimum height=3em,
		text centered,
		on chain},
	line/.style={draw, thick, <-},
	element/.style={
		tape,
		top color=white,
		bottom color=blue!50!black!60!,
		minimum width=8em,
		draw=blue!40!black!90, very thick,
		text width=10em,
		minimum height=3.5em,
		text centered,
		on chain},
	every join/.style={->, thick,shorten >=1pt},
	decoration={brace},
	tuborg/.style={decorate},
	tubnode/.style={midway, right=2pt},
}

\newif\iflong
\longtrue% set to true for long version

\newif\ifmaybe
\maybefalse% set to true for material that does not fit into main line of paper

\begin{document}

%Register automata
\newcommand{\A}{{\mathcal A}}
\newcommand{\B}{{\mathcal B}}
%Domain of data values
\newcommand{\D}{{\mathcal D}}
%Guards/formulas
\newcommand{\G}{{\mathcal G}}
%Domain of relations
\newcommand{\R}{{\mathcal R}}
%Universe of variables
\newcommand{\V}{{\mathcal V}}
%Prefixes
\newcommand{\Pref}{\mathit{Pref}}
%Valuations
\newcommand{\Val}{\mathit{Val}}
%Variables
\newcommand{\Var}{\mathit{Var}}

\newcommand{\setp}{\mathit{setp}}
\newcommand{\gain}{\mathit{gain}}
\newcommand{\sens}{\mathit{sens}}
\newcommand{\cntr}{\mathit{cntr}}
\newcommand{\reset}{\mathit{reset}}
\newcommand{\setpoint}{\mathit{sp}}
\newcommand{\sensorvalue}{\mathit{sv}}

%\newcommand{\map}{\mathit{map}}
%map
\newcommand{\nat}{\mathbb{N}}
%nat
\newcommand{\arity}{\mathit{arity}}
%conjunction of guards in symbolic trace
\newcommand{\guard}{\mathit{guard}}
%number of input symbols in symbolic trace
\newcommand{\length}{\mathit{length}}
%matching induced by register eqquivalence
\newcommand{\matching}{\mathit{matching}}
%satisfiable
\newcommand{\satisfiable}{\mathit{Sat}}
%arity

%Basic run
\newcommand{\btr}{\mathit{symb}}
%From tainted runs to runs
\newcommand{\run}{\mathit{run}}
%Symbolic run
\newcommand{\symbolic}{\mathit{symb}}
%Trace
\newcommand{\trace}{\mathit{trace}}
%Symbolic trace
\newcommand{\strace}{\mathit{strace}}
%Domain
\newcommand{\domain}{\mathit{domain}}
%Range
\newcommand{\range}{\mathit{range}}

\newtheorem{theorem}{Theorem}
\newtheorem{lemma}[theorem]{Lemma}
\newtheorem{corollary}[theorem]{Corollary}
\newtheorem{definition}[theorem]{Definition}
\newtheorem{example}[theorem]{Example}
\newproof{proof}{Proof}

\title{A Myhill-Nerode Theorem for Register Automata and Symbolic Trace Languages\tnoteref{t1}}
\tnotetext[t1]{Supported by NWO TOP project 612.001.852 Grey-box learning of Interfaces for Refactoring Legacy Software (GIRLS).}

\author[nijmegen]{Frits Vaandrager}
\ead{F.Vaandrager@cs.ru.nl}

\author[bangalore]{Abhisek Midya\fnref{Radboud}}
\ead{abhisekmidyacse@gmail.com}

\address[nijmegen]{Institute for Computing and Information Sciences, Radboud University, Toernooiveld 212, 6525 EC, Nijmegen, The Netherlands}
\address[bangalore]{Department of Information Science and Engineering, CMRIT, Bangalore, India}

\fntext[Radboud]{Work on this article was carried out while the author was employed at Radboud University.}

\begin{abstract}
	We propose a new symbolic trace semantics for register automata (extended finite state machines) which records both the sequence of input symbols that occur during a run as well as the constraints on input parameters that are imposed by this run.
	Our main result is a generalization of the classical Myhill-Nerode theorem to this symbolic setting.
	Our generalization requires the use of three relations to capture the additional structure of register automata. Location equivalence $\equiv_l$ captures that symbolic traces end in the same location, transition equivalence $\equiv_t$ captures that they share the same final transition, and  a partial equivalence relation $\equiv_r$ captures that symbolic values $v$ and $v'$ are stored in the same register after symbolic traces $w$ and $w'$, respectively.
    A symbolic language is defined to be regular if relations $\equiv_l$, $\equiv_t$ and $\equiv_r$ exist that satisfy certain conditions, in particular, they all have finite index.
    We show that the symbolic language associated to a register automaton is regular, and we construct, for each regular symbolic language, a register automaton that accepts this language.
    Our result provides a foundation for grey-box learning algorithms in settings where the constraints on data parameters can be extracted from code using e.g.\ tools for symbolic/concolic execution or tainting.
    We believe that moving to a grey-box setting is essential to overcome the scalability problems of state-of-the-art black-box learning algorithms. 
\end{abstract}


\maketitle

\section{Introduction}
Model learning (a.k.a. active automata learning) is a black-box technique which constructs state machine models of software and hardware components from information obtained by providing inputs and observing the resulting outputs. 
Model learning has been successfully used in numerous applications, for instance for
generating conformance test suites of software components \cite{HMNSBI2001},
finding mistakes in implementations of security-critical protocols \cite{FJV16,FiterauEtAl17,FH17},
learning interfaces of classes in software libraries \cite{HowarISBJ12}, and
checking that a legacy component and a refactored implementation have the same behavior \cite{SHV16}.
We refer to \cite{Vaa17,HowarS2018} for surveys and further references.

Myhill-Nerode theorems \cite{nerode58,HU79} are of pivotal importance for model learning algorithms. Angluin's classical $L^{\ast}$ algorithm \cite{Ang87} for active learning of regular languages, as well as improvements such as \cite{RivestS93,Isberner2014,ShahbazG09}, use an observation table to approximate the Nerode congruence.
Maler and Steiger \cite{MalerS97} established a Myhill-Nerode theorem for $\omega$-languages that serves as a basis for
a learning algorithm described in \cite{AngluinF16}.
The $\mathit{SL}^{\ast}$ algorithm for active learning of register automata of Cassel et al \cite{CasselHJS16} is directly based on a generalization of the classical Myhill-Nerode theorem to a setting of data languages and register automata (extended finite state machines).
Francez and Kaminski \cite{FrancezK03}, Benedikt et al \cite{BenediktLP10} and Boja\'nczyk et al \cite{BojanczykKL11}
all present Myhill-Nerode theorems for data languages.
	
Despite the convincing applications of black-box model learning, it is fair to say that existing algorithms do not scale very well.  In order to learn models of realistic applications in which inputs and outputs carry data parameters,
state-of-the-art techniques either rely on manually constructed mappers that abstract the data parameters of
inputs and outputs into a finite alphabet \cite{AJUV15}, or otherwise infer guards
and assignments from black-box observations of test outputs \cite{CasselHJS16,CEGAR12}.
The latter can be costly, especially for models where the control flow depends on data parameters in the input.
Thus, for instance, the RALib tool \cite{CasselRALib}, an implementation of the $\mathit{SL}^{\ast}$ algorithm, needed more than two hundred thousand input/reset events to learn register automata with just 6 to 8 locations for TCP client implementations of Linux, FreeBSD and Windows \cite{FH17}.  Existing black-box model learning algorithms also face severe restrictions on the operations and predicates on data that are supported (typically, only equality/inequality predicates and constants).

A natural way to address these limitations is to augment learning algorithms with white-box information extraction methods,
which are able to obtain information about the system under learning at lower cost than black-box techniques \cite{HowarJV19}. 
Constraints on data parameters can be extracted from the code using e.g.\ tools for symbolic execution \cite{SymbEx},
concolic execution \cite{DART}, or tainting \cite{HoscheleZ16}.
Several researchers have successfully explored this idea, see for instance \cite{acm2414956,acm2028077,BotincanB13,acm2483783}.
Recently, we showed how constraints on data parameters can be extracted from Python programs using tainting, and used to 
boost the performance of RALib with almost two orders of magnitude.
We were also able to learn models of systems that are completely out of reach of black-box techniques, such as ``combination locks'', systems that only exhibit certain behaviors after a very specific sequence of inputs  \cite{TaintingRALib}.
Nevertheless, all these approaches are rather ad hoc, and what is missing is Myhill-Nerode theorem for this enriched settings that may serve as a foundation for grey-box model learning algorithms for a general class of register automata.
In this article, we present such a theorem.

More specifically, we propose a new symbolic trace semantics for register automata which records both the sequence of input symbols that occur during a run as well as the constraints on input parameters that are imposed by this run.
Our main result is a Myhill-Nerode theorem for symbolic trace languages.
Whereas the original Myhill-Nerode theorem refers to a single equivalence relation $\equiv$ on words, and constructs a DFA in which states are equivalence classes of $\equiv$, our generalization requires the use of three relations to capture the additional structure of register automata. Location equivalence $\equiv_l$ captures that symbolic traces end in the same location, transition equivalence $\equiv_t$ captures that they share the same final transition, and  a partial equivalence relation $\equiv_r$ captures that symbolic values $v$ and $v'$ are stored in the same register after symbolic traces $w$ and $w'$, respectively.
A symbolic language is defined to be regular if relations $\equiv_l$, $\equiv_t$ and $\equiv_r$ exist that satisfy certain conditions, in particular, they all have finite index.
Whereas in the classical case of regular languages the Nerode equivalence $\equiv$ is uniquely determined, different relations relations $\equiv_l$, $\equiv_t$ and $\equiv_r$ may exist that satisfy the conditions for 
regularity for symbolic languages.
We show that the symbolic language associated to a register automaton is regular, and we construct, for each regular symbolic language, a register automaton that accepts this language. In this automaton, the locations are equivalence classes of $\equiv_l$, the transitions are equivalence classes of $\equiv_t$, and the registers are
equivalence classes of $\equiv_r$. In this way, we obtain a natural generalization of the classical Myhill-Nerode theorem for symbolic languages and register automata. Unlike Cassel et al \cite{CasselHJS16}, we need no restrictions on the allowed data predicates to prove our result, which drastically increases the range of potential applications.
Our result paves the way for efficient grey-box learning algorithms in settings where the constraints on data parameters can be extracted from the code.

\iflong
\else
Due to the page limit, proofs have been omitted from this article, except for outlines of the proofs of main Theorems ~\ref{theorem soundness} and~\ref{theorem completeness}. All proofs can be found in the report version on arXiv \cite{VM20}.
\fi

\section{Preliminaries}
\par In this section, we fix some basic vocabulary for (partial) functions, languages, and logical formulas.

\subsection{Functions}
We write $f : X \rightharpoonup Y$ to denote that $f$ is a partial function from set $X$ to
set $Y$.  For $x \in X$, we write $f(x) \downarrow$ if there exists a $y \in Y$ such that $f(x)=y$, i.e., the result
is defined, and $f(x) \uparrow$ if the result is undefined.
We write $\domain(f) = \{ x \in X \mid f(x)\downarrow \}$ and $\range(f) = \{ f(x) \in Y \mid x \in \domain(f) \}$.
We often identify a partial function $f$ with the set of pairs $\{ (x,y) \in X \times Y \mid f(x)=y \}$.
As usual, we write $f : X \rightarrow Y$ to denote that $f$ is a total function from $X$ to $Y$, that is, $f : X \rightharpoonup Y$ and $\domain(f) = X$.

\subsection{Languages}

Let $\Sigma$ be a set of \emph{symbols}.
A \emph{word} $u = a_1 \ldots a_n$ over $\Sigma$ is a finite sequence of symbols from $\Sigma$.
The \emph{length} of a word $u$, denoted $|u|$ is the number of symbols occurring in it.
The empty word is denoted $\epsilon$.
We denote by $\Sigma^{\ast}$ the set of all words over $\Sigma$.
Given two words $u$ and $w$, we denote by $u \cdot w$ the concatenation of $u$ and $w$.
When the context allows it, $u \cdot w$ shall be simply written $u w$.
We say that $u$ is a \emph{prefix} of $w$ iff there exists a word $u'$ such that $u \cdot u' = w$.
Similarly, $u$ is a \emph{suffix} of $w$ iff there exists a word $u'$ such that $u' \cdot u  = w$.
A \emph{language} $L$ over $\Sigma$ is any set of words over $\Sigma$, so therefore a subset of $\Sigma^{\ast}$.
We say that $L$ is prefix closed if, for each $w \in L$ and each prefix $u$ of $w$, $u \in L$ as well.
%If $L$ is a language over $\Sigma$, then the \emph{prefix closure} of $L$, denoted $\Pref(L)$, is the set of all prefixes of words in $L$.

\subsection{Guards}
We postulate a countably infinite set $\V = \{ v_1, v_2,\ldots \}$ of \emph{variables}.
In addition, there is also a variable $p \not\in\V$ that will play a special role as formal parameter of input symbols; we write $\V_p = \V \cup \{ p \}$.
Our framework is parametrized by a set $R$ of relation symbols. Elements of $R$ are assigned finite \emph{arities}.
A {\em guard} is a Boolean combination of relation symbols from $R$ over variables.
Formally, the set of \emph{guards} is inductively defined as follows:
\begin{itemize}
	\item 
	$\top$ is a guard.
    \item 
    If $r\in R$ is an $n$-ary relation symbol and $x_1 ,\ldots, x_n$ are variables from $\V_p$, then $r(x_1,\ldots,x_n)$ is a guard.
    \item
    If $g$ is a guard then $\neg g$ is a guard.
    \item
    If $g_1$ and $g_2$ are guards then $g_1 \wedge g_2$ is a guard.
\end{itemize}
We use standard abbreviations from propositional logic such as $g_1 \vee g_2$.
We write $\Var(g)$ for the set of variables that occur in a guard $g$.
We say that $g$ is a guard \emph{over} set of variables $X$ if $\Var(g) \subseteq X$.
We write $\G (X)$ for the set of guards over $X$, and use symbol $\equiv$ to denote syntactic equality of guards.

We postulate a structure $\R$ consisting of a set $\D$ of \emph{data values} and a distinguished $n$-ary relation $r^{\R} \subseteq \D^n$ for each $n$-ary relation symbol $r \in R$.
In a trivial example of a structure $\R$, $R$ consists of the binary symbol `$=$', $\D$ the set of natural numbers, and
$=^{\R}$ is the equality predicate on numbers.
An $n$-ary operation $f : \D^n \rightarrow \D$ can be modelled in our framework as a predicate of arity $n+1$.
We may for instance extend structure $\R$ with a ternary predicate symbol $+$, where $(d_1, d_2, d_3) \in +^{\R}$ iff the sum of $d_1$ and $d_2$ equals $d_3$. Constants like $0$ and $1$ can be added to $\R$ as unary predicates.

A \emph{valuation} is a partial function $\xi : \V_p \rightharpoonup \D$ that assigns data values to variables.
%We write $\Val(X)$ for the set of valuations for $X$.
If $\Var(g) \subseteq\domain(\xi)$, then $\xi \models g$ is defined inductively by:
\begin{itemize}
	\item 
	$\xi \models \top$
	\item 
	$\xi \models r(x_1,\ldots,x_n)$  iff $(\xi(x_1),\ldots,\xi(x_n)) \in r^{\R}$
	\item 
	$\xi \models \neg g$  iff not $\xi \models g$
	\item 
	$\xi \models g_1 \wedge g_2$ iff $\xi \models g_1$ and $\xi \models g_2$
\end{itemize}
If $\xi \models g$ then we say valuation $\xi$ \emph{satisfies} guard $g$.
We call $g$ is \emph{satisfiable}, and write $\satisfiable(g)$, if there exists a valuation $\xi$ such that $\xi\models g$.
Guard $g$ is a \emph{tautology} if $\xi\models g$ for all valuations $\xi$ with $\Var(g) \subseteq\domain(\xi)$.
%
%We represent (partial) functions as sets of pairs, and write $X \nrightarrow Y$ for the set of partial functions
%from $X$ to $Y$.

A \emph{variable renaming} is a partial function $\sigma: \V_p \rightharpoonup \V_p$.
If $g$ is a guard with $\Var(g) \subseteq \domain(\sigma)$ then $g[\sigma]$ is the guard obtained by replacing each occurrence of a variable $x$ in $g$ by variable $\sigma(x)$.
%Note that $\xi \circ \sigma \models g$ iff $\xi \models g[\sigma]$.
The following lemma is easily proved by induction.

\begin{lemma}
\label{lemma variable renaming}
$\xi \circ \sigma \models g$ iff $\xi \models g[\sigma]$
\end{lemma}
\iflong
\begin{proof}
By induction on structure of $g$:
\begin{itemize}
	\item 
	$g \equiv \top$: Statement follows because $\xi \circ \sigma \models \top$ and, as $\top[\sigma] \equiv \top$, $\xi \models \top[\sigma]$.
	\item 
	$g \equiv r(x_1,\ldots, x_n)$:
	\begin{eqnarray*}
	\xi \circ \sigma \models g & \Leftrightarrow & (\xi \circ \sigma(x_1),\ldots,\xi \circ\sigma (x_n)) \in r^{\R}\\
	& \Leftrightarrow & (\xi ( \sigma(x_1)),\ldots,\xi (\sigma (x_n))) \in r^{\R}\\
	& \Leftrightarrow & \xi \models r(\sigma(x_1),\ldots,\sigma(x_n))\\
	& \Leftrightarrow & \xi \models g[\sigma]
	\end{eqnarray*}
		\item 
	$g \equiv \neg g'$:
	\begin{eqnarray*}
		\xi \circ \sigma \models g & \Leftrightarrow &   \mbox{ not } \xi \circ \sigma \models g'\\
		& \Leftrightarrow & \mbox{ not } \xi \models g'[\sigma] \mbox{ (by induction hypothesis)}\\
		& \Leftrightarrow & \xi \models \neg g'[\sigma]\\
		& \Leftrightarrow & \xi \models g[\sigma]
	\end{eqnarray*}
\item 
$g \equiv g_1 \wedge g_2$:
\begin{eqnarray*}
	\xi \circ \sigma \models g & \Leftrightarrow & \xi \circ \sigma \models g_1 \mbox{ and } \xi \circ \sigma \models g_2 \mbox{ (by induction hypothesis)}\\
	& \Leftrightarrow & \xi \models g_1[\sigma] \mbox{ and } \xi \models g_2 [\sigma]\\
	& \Leftrightarrow & \xi \models g_1[\sigma] \wedge g_2[\sigma]\\
	& \Leftrightarrow & \xi \models g[\sigma]  \hspace{7cm} \square
\end{eqnarray*}
\end{itemize}
\end{proof}
\fi

\section{Register Automata}
In this section, we introduce register automata and show how they may be used as recognizers for both data languages and symbolic languages. 

\subsection{Definition and trace semantics}
A register automaton comprises a set of locations with transitions between them, and a set of registers which can store data values that are received as inputs. Transitions contain guards over the registers and the current input, and may
assign new values to registers.

\begin{definition}
A \emph{register automaton} is a tuple $\A = (\Sigma, Q, q_0, V, \Gamma)$, where
\begin{itemize}
\item
$\Sigma$ is a finite set of \emph{input symbols},
\item $Q$ is a finite set of \emph{locations}, with $q_0 \in Q$ the \emph{initial location},
\item $V \subset \V$ is a finite set of \emph{registers}, and 
\item $\Gamma$ is a finite set of \emph{transitions}, each of form $\langle q, \alpha, g, \varrho, q'\rangle$ where 
\begin{itemize}
\item $q,q'\in Q$ are the \emph{source} and \emph{target} locations, respectively,
\item $\alpha \in\Sigma$ is an input symbol,
\item $g \in \G(V \cup \{ p \})$ is a guard, and
\item $\varrho : V \rightharpoonup V\cup \{ p \}$ is an \emph{assignment}; we require that $\varrho$ is injective.
\end{itemize}
\end{itemize}
Register automata are required to be {\em deterministic} in the sense
that for each location $q \in Q$ and input symbol $\alpha \in \Sigma$,
the conjunction of the guards of any pair of distinct $\alpha$-transitions with source $q$ is not satisfiable.
We write $q \xrightarrow{\alpha, g, \varrho} q'$ if $\langle q, \alpha, g, \varrho, q'\rangle \in \Gamma$.
\end{definition}

\begin{example}
\label{running example}
Figure~\ref{fig: register automaton} shows a register automaton $\A = (\Sigma, Q, q_0, V, \Gamma)$ with
a single input symbol $a$ and three locations $q_0$, $q_1$ and $q_2$.
\begin{figure}[h!]
\begin{center}
\begin{tikzpicture}[shorten >=1pt,node distance=3.3cm,on grid,auto] 
   \node[state,initial,accepting] (q_0)   {$q_0$}; 
   \node[state,accepting] (q_1) [ right=of q_0] {$q_1$}; 
   \node[state,accepting](q_2) [right=of q_1] {$q_2$};
    \path[->]
    (q_0) edge node {$a, x:=p$} (q_1)
    (q_1) edge [loop above] node {$a,  x\leq p, x:=p$} ()
    (q_1) edge  node {$a, p<x, x:=p$} (q_2)
    (q_2) edge [bend left] node {$a, x\leq p, x:=p$} (q_1);
\end{tikzpicture}
\caption{Register automaton.}
\label{fig: register automaton}
\end{center}
\end{figure}
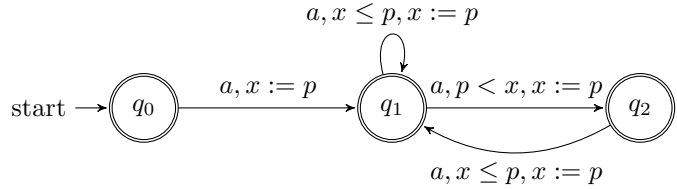
The initial location $q_0$ is marked by an arrow ``start''. 
There is just a single register $x$. Set $\Gamma$ contains four transitions, which are indicated in the diagram.
All transitions are labeled with input symbol $a$, a guard over formal parameter $p$ and the registers, and an assignment. Guards represent conditions on data values. For example, the guard on the transition from $q_1$ to $q_2$, expresses that the data value of action $a$ must be smaller than the data value currently stored in register $x$. We write $x:= p$ to denote the assignment that stores the data parameter $p$ in register $x$, that is, the function $\varrho$ satisfying $\varrho(x) = p$.  
Note that in location $q_1$, which has more than one outgoing transition, the conjunction is not satisfiable.
In the graphical representation of register automata, trivial guards ($\top$) and assignments (empty domain) are omitted.
\end{example}

\begin{example}
\label{controller example}
The register automaton of Figure~\ref{fig: controller} describes a simple proportional controller.
In such a controller, the output is proportional to the error signal, which is the difference between the set point and the value reported by the sensor.
In the first transition from initial location $q_0$, the controller receives the set point and stores it in register $\setpoint$.
In the next transition from location $q_1$, the controller receives the value for the proportional gain and stores it in register $K$.
Now whenever the controller receives a sensor value, it stores this value in register $\sensorvalue$, and then outputs $K * (\setpoint - \sensorvalue)$,
the product of proportional gain and the error signal.
However, due to physical limitations of the actuator, the output is bounded between $-30$ and $30$.
Whenever $K * (\setpoint - \sensorvalue)$ lies outside this interval, the controller sets the output to either $-30$ or $30$,
and returns to its initial state via a reset transition.
	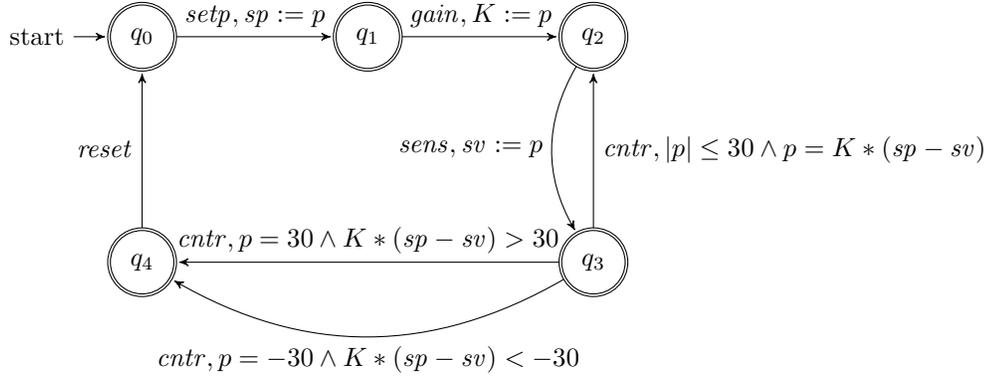
\begin{figure}[h!]
		\begin{center}
			\begin{tikzpicture}[shorten >=1pt,node distance=3cm,on grid,auto] 
			\node[state,initial,accepting] (q_0)   {$q_0$}; 
			\node[state,accepting] (q_1) [ right=of q_0] {$q_1$}; 
			\node[state,accepting](q_2) [right=of q_1] {$q_2$};
			\node[state,accepting](q_3) [below=of q_2] {$q_3$};
			\node[state,accepting](q_4) [below=of q_0] {$q_4$};
			\path[->]
			(q_0) edge node {$\setp,  \setpoint :=p$} (q_1)
			(q_1) edge node {$\gain, K := p$} (q_2)
			(q_2) edge [bend right] node[left] {$\sens, \sensorvalue :=p$} (q_3)
			(q_3) edge  node[right] {$\cntr, |p| \leq  30 \wedge  p = K * (\setpoint - \sensorvalue)$} (q_2)
			(q_3) edge node[above]  {$\cntr,p=30 \wedge  K * (\setpoint - \sensorvalue) > 30$} (q_4)
			(q_3) edge [bend left] node {$\cntr,p=-30 \wedge  K * (\setpoint - \sensorvalue) < -30$} (q_4)
			(q_4) edge node {$\reset$} (q_0);
			\end{tikzpicture}
			\caption{A register automaton model of a proportional controller.}
			\label{fig: controller}
		\end{center}
	\end{figure}

The register automaton of Figure~\ref{fig: register automaton} can be easily expressed within the input language of the black box learning tool RALib \cite{CasselRALib},
and in fact this tool is able to learn this automaton. The relations/predicates used in the register automaton of Figure~\ref{fig: controller} are more complicated, and black box learning of this automaton is beyond the capabilities of state-of-the-art active learning algorithms and tools.
\end{example}

The semantics of a register automaton is defined in terms of the set of \emph{data words} that it accepts.

\begin{definition}
\label{def data language}
Let $\Sigma$ be a finite alphabet.
A \emph{data symbol} over $\Sigma$ is a pair $\alpha(d)$ with $\alpha \in \Sigma$ and $d \in \D$.
A \emph{data word} over $\Sigma$ is a finite sequence of data symbols, i.e., a word over $\Sigma \times \D$.
A \emph{data language} over $\Sigma$ is a set of data words over $\Sigma$.
\end{definition}

We associate a data language to each register automata as follows.

\begin{definition}
\label{def semantics register automata}
Let $\A = (\Sigma, Q, q_0, V, \Gamma)$ be a register automaton.
A \emph{configuration} of $\A$ is a pair $(q, \xi)$, where $q \in Q$ and $\xi: V \rightharpoonup \D$.
%%%
A \emph{run of $\A$ over} a data word $w = \alpha_1 (d_1)\cdots\alpha_n (d_n)$ is a sequence
\begin{eqnarray*}
\gamma & = & (q_0, \xi_0) ~ \xrightarrow{\alpha_1(d_1)} ~ 
(q_1, \xi_1) \quad
\ldots \quad
(q_{n-1}, \xi_{n-1}) ~ \xrightarrow{\alpha_n(d_n)} ~ (q_n, \xi_n ),
\end{eqnarray*}
where, for $0 \leq i \leq n$, $(q_i, \xi_i)$ is a configuration of $\A$, $\domain(\xi_0) = \emptyset$, and for $0 < i \leq n$, $\Gamma$ contains a transition $q_{i-1} \xrightarrow{\alpha_i, g_i, \varrho_i} q_i$ such that
\begin{itemize}
    \item
    $\iota_i \models g_i$, where $\iota_i = \xi_{i-1} \cup \{ (p, d_i) \}$, and
    \item
    $\xi_i = \iota_i \circ \varrho_i$.
\end{itemize} 
We call $w$ the \emph{trace} of $\gamma$, notation $\trace(\gamma) = w$.
Data word $w$ is \emph{accepted} by $\A$ if $\A$ has a run over $w$.
The \emph{data language} of $\A$, notation $L(\A)$, is the set of all data words that are accepted by $\A$.
Two register automata with the same sets of input symbols are trace equivalent if they accept the same data language.
\end{definition}

\begin{example}
\label{example run}
Consider the register automaton of Figure~\ref{fig: register automaton}.
This automaton accepts the data word  $a(1)~ a(4)~ a(0) ~ a(7)$ since
the following sequence of steps is a run (here $\xi_0$ is the trivial function with empty domain):
\[
(q_0, \xi_0) \xrightarrow{a(1)} (q_1, x \mapsto 1) \xrightarrow{a(4)} (q_1, x \mapsto 4) \xrightarrow{a(0)} (q_2, x \mapsto 0) \xrightarrow{a(7)} (q_1, x \mapsto 7).
\]
Upon receiving the first input $a(1)$, the automaton jumps to $q_1$ and stores data value $1$ in the register $x$. Since $4$ is bigger than
$1$, the automaton takes the self loop upon receiving the second input $a(4)$ and stores $4$ in $x$. Since $0$ is less than $4$, it moves to $q_2$ upon receipt of the third input $a(0)$ and updates $x$ to $0$. Finally, the automaton gets back to $q_1$ as $7$ is bigger than $0$.
\end{example}

Suppose that in the register automaton of Figure~\ref{fig: register automaton} we replace the guard on the transition from $q_0$ to $q_1$ by $x \leq p$.  Since initial valuation $\xi_0$ does not assign a value to $x$, this means that it is not defined whether $\xi_0$ satisfies guard $x \leq p$. Automata in which such ``runtime errors'' do not occur are called \emph{well-formed}.

\begin{definition}
	Let $\A$ be a register automaton.
	We say that a configuration $(q, \xi)$ of $\A$ is \emph{reachable} if there is a run of $\A$ that ends with $(q, \xi)$.   We call $\A$ \emph{well-formed} if, for each reachable configuration $(q, \xi)$, $\xi$ assigns a value to all variables from $V$ that occur in guards of outgoing transitions of $q$, that is, 
	\begin{eqnarray*}
		(q, \xi) \mbox{ reachable } \wedge q \xrightarrow{\alpha, g \varrho} q' & \Rightarrow & \Var(g) \subseteq \domain(\xi) \cup \{ p \}. \label{variables guards always defined}
	\end{eqnarray*}
\end{definition}
As soon as the set of data values and the  collection of predicates becomes nontrivial, well-formedness of register
automata becomes undecidable.
However, it is easy to come up with a sufficient condition for well-formedness, based on a syntactic analysis of $\A$,
which covers the cases that occur in practice.
In the remainder of article, we will restrict our attention to well-formed register automata. In particular, the register automata that are constructed from regular symbolic trace languages in our Myhill-Nerode theorem will be well-formed.

\paragraph{Relation with automata of Cassel at al.}
	Our definition of a register automaton is different from the one used in the $\mathit{SL}^{\ast}$ algorithm of Cassel et al \cite{CasselHJS16} and its implementation in RALib \cite{CasselRALib}. It is instructive to compare the two definitions.
\begin{enumerate}
	\item
		In order to establish a Myhill-Nerode theorem, \cite{CasselHJS16} requires that structure $\R$, which is a parameter of the $\mathit{SL}^{\ast}$ algorithm, is \emph{weakly extendible}. This technical restriction excludes many data types that are commonly used in practice. For instance, the set of integers with constants $0$ and $1$, an addition operator $+$, and a less-than predicate $<$ is not weakly extendable. 
		For readers familiar with \cite{CasselHJS16}: a structure (called theory in \cite{CasselHJS16}) is weakly extendable if for all natural numbers $k$ and data words $u$, there exists a $u'$ with $u' \approx_{\R} u$ which is $k$-extendable. Intuitively, $u' \approx_{\R} u$ if data words $u'$ and $u$ have the same sequences of actions and cannot be distinguished by the relations in $\R$. Let 
		\begin{eqnarray*}
		u & = &  \alpha(0) \;  \alpha(1) \;  \alpha(2) \; \alpha(4) \; \alpha(8) \; \alpha(16) \; \alpha(11). 
		\end{eqnarray*}
		Then there exists just one $u'$ different from $u$ with $u' \approx_{\R} u$, namely 
		\begin{eqnarray*}
		u' & = &  \alpha(0)\;  \alpha(1) \;  \alpha(2) \; \alpha(4) \; \alpha(8) \; \alpha(16) \; \alpha(13). 
		\end{eqnarray*}
	(Since $0$ and $1$ are constants, the first two data parameters must be equal. Since $+(1,1,2)$ also the third data parameters must be equal, etc. Since $8 < 11 < 16$, the final data parameter must be in between $8$ and $16$, but we cannot pick $9$, $10$, $12$, $14$ and $15$ because
	$+(1, 8, 9)$, $+(2,8,10)$, $+(4,8,12)$, $+(2, 14, 16)$ and $+(1,15,16)$, respectively.)
		Now both $u$ and $u'$ are not even $1$-extendable: if we extend $u$ with $\alpha(3)$, we cannot find a matching extension $\alpha(d')$ of $u'$ such that $u \; \alpha(3) \approx_{\R} u' \; \alpha(d')$, and if we extend $u'$ with $\alpha(5)$ we cannot find a matching extension $\alpha(d)$ of $u$ such that $u \; \alpha(d) \approx_{\R} u' \; \alpha(5)$.
	In the terminology of model theory \cite{Poizat}, a structure is $k$-extendable if the Duplicator can win certain $k$-move Ehrenfeucht-Fra\"{\i}ss\'e games. For structures $\R$ that are homogeneous, one can always win these games, for all $k$. Thus, homogeneous structures are weakly extendible.
	An even stronger requirement, which is imposed in work of \cite{MoermanEtAl17} on nominal automata, is that $\R$ is $\omega$-categorical. 
	 In our approach, no restrictions on $\R$ are needed.
		\item 
		 Unlike \cite{CasselHJS16}, we do not associate a fixed set of variables to each location.  Our definition is slightly more general, which simplifies some technicalities.
		\item
		However, we require assignments to be injective, a restriction that is not imposed by \cite{CasselHJS16}. But note that the register automata that are actually constructed by $\mathit{SL}^{\ast}$ are \emph{right-invariant} \cite{CasselHJMS15}. In a right-invariant register automaton, two values can only be tested for equality if one of them is the current input symbol.  Right-invariance, as defined in \cite{CasselHJMS15}, implies that assignments are injective. As illustrated by the example of Figure~\ref{fig:succinct}, our register automata are exponentially more succinct than the right-invariant register automata constructed by $\mathit{SL}^{\ast}$. As pointed out in \cite{CasselHJMS15}, right-invariant register automata in turn are more succinct than the automata of \cite{FrancezK03,BenediktLP10}.
\begin{figure}[htb!]
	\begin{center}
		\begin{tikzpicture}[shorten >=1pt,node distance=1.7cm,on grid,auto] 
		\node[state,initial,accepting] (q_0)   {$q_0$}; 
		\node[state,accepting] (q_1) [ right=of q_0] {$q_1$}; 
		\node[state,accepting](q_2) [right=of q_1] {$q_2$};
		\node[state,accepting] (q_3) [right=of q_2] {{\tiny $q_{2n-1}$}};
		\node[state,accepting] (q_4) [right=of q_3] {$q_{2n}$};
		\node[state,accepting] (q_5) [below=of q_4] {$q_{\mathit{ok}}$};
		\path[->]
		(q_0) edge[bend left] node {$a, x_1 :=p$} (q_1)
		(q_1) edge[bend left] node {$a, x_2 :=p$} (q_2)
		(q_2) edge[color=white] node[color=black] {$\cdots$} (q_3)
		(q_3) edge[bend left] node {$a, x_{2n} :=p$} (q_4)
		(q_4) edge node[left] {$b, \bigwedge_{i=1}^{n-1} ( x_i = x_{i+1} \leftrightarrow x_{n+i} = x_{n+i+1})$} (q_5)
		;
		\end{tikzpicture}
		\caption{For each $n>0$, $\A_n$ is a register automaton that first accepts $2n$ input symbols $a$, storing all the data values that it receives, and then accepts input symbol $b$ when two consecutive values in the first half of the input are equal iff the corresponding consecutive values in the second half of the input are equal.  The number of locations and transitions of $\A_n$ grows linearly with $n$.  There exist right-invariant register automata $\B_n$ that accept the same data languages, but their size grows exponentially with $n$.}
		\label{fig:succinct}
	\end{center}
\end{figure}
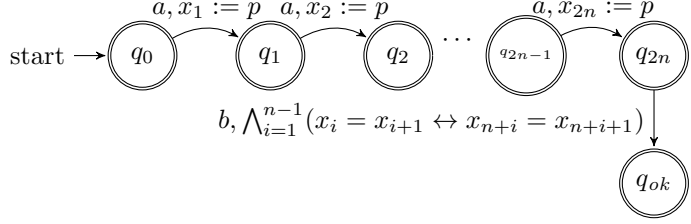
		\item 
		Since any prefix of a run is also a run, the data language accepted by a register automaton is prefix closed, a restriction that is not imposed
		in \cite{CasselHJS16}. Prefix closedness is convenient for technical reasons, but for reactive systems the restriction is actually quite natural. RALib \cite{CasselRALib} also assumes that data languages are prefix closed.
	\end{enumerate}

Since register automata are deterministic, there exists a one-to-one correspondence between accepted data words and runs.
From every run $\gamma$ of a register automaton $\A$ we can trivially extract a data word $\trace(\gamma)$ by forgetting all information except the data symbols.
Conversely, for each data word $w$ that is accepted by $\A$, there exists a corresponding run $\gamma$, which is uniquely determined by the data word since from each configuration $(q, \xi)$ and data symbol $\alpha(d)$, exactly one transition will be enabled.

\begin{lemma}
	\label{lem for each trace there is a unique run}
	Suppose $\gamma$ and $\gamma'$ are runs of a register automaton $\A$ such that $\trace(\gamma) = \trace(\gamma')$.
	Then $\gamma = \gamma'$.
\end{lemma}
\iflong
\begin{proof}
	We prove the lemma by contradiction. Suppose $\gamma \neq \gamma'$.
	All runs of $\A$ share at least the initial configuration $(q_0, \xi_0)$.
	Let $\gamma$ be as in Definition~\ref{def semantics register automata}, and let $(q_{i-1}, \xi_{i-1})$ be the last point where $\gamma$ and $\gamma'$ coincide.
	From this point, $\gamma$ continues its course with a step
	\[
	(q_{i-1}, \xi_{i-1}) ~ \xrightarrow{\alpha_i(d_i)} ~ (q_i, \xi_i),
	\]
	whereas $\gamma'$ continues with a different step
	\[
	(q_{i-1}, \xi_{i-1}) ~ \xrightarrow{\alpha_i(d_i)} ~ (q'_i, \xi'_i).
	\]
	Note that both steps carry the same data symbol as $\trace(\gamma) = \trace(\gamma')$. Then
	$\Gamma$ contains a transition $q_{i-1} \xrightarrow{\alpha_i, g_i, \varrho_i} q_i$ such that
	$\iota_i \models g_i$, where $\iota_i = \xi_{i-1} \cup \{ (p, d_i) \}$.
	In addition, $\Gamma$ contains a transition $q_{i-1} \xrightarrow{\alpha_i, g'_i, \varrho'_i} q'_i$ such that
	$\iota_i \models g'_i$.
	Since $\iota_i \models g_i$ and $\iota_i \models g'_i$, we may conclude that $g_i \wedge g'_i$ is satisfiable. Therefore, as $\A$ is deterministic, both transitions
	are the same, that is, $g_i \equiv g'_i$, $\varrho_i = \varrho'_i$ and $q_i = q'_i$.
	But then also $\xi_i = \iota_i \circ \varrho_i =  \iota_i \circ \varrho'_i = \xi'_i$, which means that the two outgoing steps of configuration $(q_{i-1}, \xi_{i-1})$ are the same. Contradiction. \qed
\end{proof}
\fi

\subsection{Symbolic semantics}
We will now introduce an alternative trace semantics for register automata, which records both the sequence of input symbols that occur during a run as well as the constraints on input parameters that are imposed by this run.  We will explore some basic properties of this semantics, and show that the equivalence induced by symbolic traces is finer than data equivalence.

A symbolic language consists of words in which input symbols and guards alternate.

\begin{definition}
%	\label{def tainted languages}
	\label{def symbolic languages}
	Let $\Sigma$ be a finite alphabet.
	A \emph{symbolic word} over $\Sigma$ is a finite alternating sequence $w = \alpha_1 G_1 \cdots \alpha_n G_n$ of input symbols from $\Sigma$ and guards.
	A \emph{symbolic language} over $\Sigma$ is a set of symbolic words over $\Sigma$.
\end{definition}

A symbolic run is just a run, except that the valuations do not return concrete data values, but markers (variables) that record the exact place in the run where the input occurred. We use variable $v_i$ as a marker for the $i$-th input value. Using these symbolic valuations (variable renamings, actually) it is straightforward to compute the constraints on the input parameters from the guards occurring in the run.

\begin{definition}
\label{def symbolic semantics}
Let $\A = (\Sigma, Q, q_0, V, \Gamma)$ be a register automaton.
A \emph{symbolic run} of $\A$ is a sequence 
\[
\delta ~=~ (q_0, \zeta_0) ~ \xrightarrow{\alpha_1, g_1, \varrho_1} ~ 
(q_1, \zeta_1) ~
\ldots ~ \xrightarrow{\alpha_n, g_n, \varrho_n} ~
(q_n, \zeta_n),
\]
where $\zeta_0$ is the trivial variable renaming with empty domain and, for $0 < i \leq n$,
\begin{itemize}
    \item 
    $q_{i-1} \xrightarrow{\alpha_i, g_i, \varrho_i} q_i$ is a transition in $\Gamma$,
    \item
    $\zeta_i$ is a variable renaming with $\domain(\zeta_i) \subseteq V$, and
    \item
    $\zeta_i = \iota_i \circ \varrho_i$, where $\iota_i = \zeta_{i-1} \cup \{ (p, v_i) \}$.
\end{itemize}
We also require that $G_1 \wedge \cdots \wedge G_n$ is satisfiable, where $G_i \equiv g_i[\iota_i]$, for $0 < i \leq n$.

The \emph{symbolic trace} of $\delta$ is the symbolic word $\strace(\delta) = \alpha_1 ~ G_1 \cdots \alpha_n ~ G_n$.
Symbolic word $w$ is \emph{accepted} by $\A$ if $\A$ has a symbolic run $\delta$ with $\strace(\delta) = w$.
The symbolic language of $\A$, notation $L_s(\A)$, is the set of all symbolic words accepted by $\A$.
Two register automata with the same sets of input symbols are symbolic trace equivalent if they accept the same symbolic language.
\end{definition}

\begin{example}
Consider the register automaton of Figure~\ref{fig: register automaton}.
The following sequence constitutes a symbolic run:
\[
(q_0, \zeta_0) \xrightarrow{a, \top, x:=p} (q_1, x \mapsto v_1) \xrightarrow{a, x \leq p, x:=p} (q_1, x \mapsto v_2)  
\]
\[ 
~~~~~~~~~~~~ \xrightarrow{a, p < x, x:=p} (q_2, x \mapsto v_3) \xrightarrow{a, x \leq p, x:=p} (q_1, x \mapsto v_4).
\]
Since
\begin{eqnarray*}
	\iota_1 & = &  \{ (p, v_1) \}\\
	\iota_2 & = & \{ (x, v_1), (p, v_2) \}\\
	\iota_3 & = & \{ (x, v_2), (p, v_3)\}\\
	\iota_4 & = & \{ (x, v_3), (p, v_4) \},
\end{eqnarray*}
the automaton accepts the symbolic word  $w = a ~ \top ~ a ~ v_1 \leq v_2 ~ a ~ v_3 < v_2 ~ a ~ v_3 \leq v_4$.
Note that the guard of $w$ is satisfiable, for instance by the valuation $\xi$ with
$\xi(v_1) = 1$, $\xi(v_2) = 4$, $\xi(v_3) = 0$ and $\xi(v_4) = 7$, which
corresponds to the (concrete) run of Example~\ref{example run}.
\end{example}

\begin{example}
	Consider the proportional controller of Figure~\ref{fig: controller}. A data word accepted by this register automaton is:
	\[
	\setp(10) ~ \gain(0,5) ~ \sens(20) ~ \cntr(-5) ~ \sens(80) ~ \cntr(-30) ~ \reset(0).
	\]
	The corresponding symbolic word is:
	\begin{eqnarray*}
	\setp ~ \top ~ \gain ~ \top & \sens ~ \top &  \cntr ~ |v_4| \leq  30 \wedge  v_4 = v_2 * (v_1 - v_3) \\
	&  \sens ~ \top & \cntr ~ v_6 =-30 \wedge  v_2  * (v_1 - v_5) < -30 ~ \reset ~ \top.
	\end{eqnarray*}
\end{example}

The two technical lemmas below state some basic properties about variable renamings in a symbolic run. The proofs are straightforward, by induction.

\begin{lemma}
\label{lemma symbolic run}
Let $\delta$ be a symbolic run of $\A$, as in Definition~\ref{def symbolic semantics}.
Then $\range(\zeta_i) \subseteq \{ v_1,\ldots, v_i \}$, for $i \in \{ 0,\ldots,n \}$,
and $\range(\iota_i) \subseteq \{ v_1,\ldots, v_i \}$, for $i \in \{ 1,\ldots,n \}$.
\end{lemma}
\iflong
\begin{proof}
By induction on $i$:
\begin{itemize}
    \item
    Base. Suppose $i=0$. Then the lemma holds trivially since $\range( \zeta_0) = \emptyset$.
    \item
    Induction step. Suppose $i>0$. Then
    \begin{eqnarray*}
    \range(\zeta_i) & = &  \range(\iota_i \circ \varrho_i) \\
    &\subseteq& \range(\iota_i) \\
    & = & \range(\zeta_{i-1} \cup \{ (p, v_i) \}) \\
    &=& \range(\zeta_{i-1}) \cup \{ v_i \} \mbox{ (by induction hypothesis) }\\
    &\subseteq & \{ v_1,\ldots, v_{i-1} \} \cup \{ v_i \}
    = \{ v_1,\ldots, v_i \}.
    \end{eqnarray*}
\end{itemize}
\qed
\end{proof}
\fi

As a consequence of our assumption that assignments in a register automaton are injective, all the variable renamings in a symbolic run are injective as well.

\begin{lemma}
\label{lemma injective valuations}
Let $\delta$ be a symbolic run of $\A$, as in Definition~\ref{def symbolic semantics}.
Then, for each $i \in \{ 0,\ldots, n \}$, $\zeta_i$ is injective, and for each $i \in \{ 1,\ldots, n\}$, $\iota_i$ is injective.
\end{lemma}
\iflong
\begin{proof}
By induction on $i$:
\begin{itemize}
    \item 
    Base. Suppose $i=0$. Then the lemma holds trivially since $\range( \zeta_0) = \emptyset$.
    \item
    Induction step. Suppose $i>0$. By Lemma~\ref{lemma symbolic run}, $\range(\zeta_{i-1}) \subseteq \{ v_1,\ldots, v_{i-1} \}$.
    By the induction hypothesis, $\zeta_{i-1}$ is injective. From this we conclude that $\zeta_{i-1} \cup \{ (p, v_i) \}$ is injective,
    which means $\iota_i$ is injective.
    Since the composition of injective functions is injective, $\zeta_i  =  \iota_i \circ \varrho_i$ is injective. \qed
\end{itemize}
\end{proof}
\fi

All symbolic words accepted by a register automaton satisfy some basic sanity properties: guards may only refer to the markers for values received thus far, and the conjunction of all the guards is satisfiable. We call symbolic words that satisfy these properties \emph{feasible}. Note that if a symbolic word is feasible, any prefix is feasible as well.

\begin{definition}[Feasible]
	Let $w = \alpha_1 G_1 \cdots \alpha_n G_n$ be a symbolic word. We write $\length(w) = n$ and  $\guard(w) = G_1 \wedge \cdots \wedge G_n$.
	Word $w$ is \emph{feasible} if $\guard(w)$ is satisfiable and $\Var(G_i) \subseteq \{ v_1,\ldots, v_i \}$, for each $i \in \{ 1,\ldots, n\}$.
	A symbolic language is \emph{feasible} if it is nonempty, prefix closed and consists of feasible symbolic words.
\end{definition}

\begin{lemma}
	\label{lemma tainted language register automaton is feasible}
	$L_s(\A)$ is feasible.
\end{lemma}
\iflong
\begin{proof}
	Since the initial configuration $(q_0, \zeta_0)$ is a symbolic run, the empty word $\epsilon$ is a symbolic word of $\A$, and so $L_s(\A)$ is nonempty.
	Since a prefix of a symbolic run is a symbolic run,	$L_s(\A)$ is prefix closed.
	Suppose $w = \alpha_1 G_1 \cdots \alpha_n G_n$ is a symbolic word of $\A$.  It suffices to show that $w$ is feasible.
	Consider a symbolic run $\delta$ for $w$, as in Definition~\ref{def symbolic semantics}.
	By Lemma~\ref{lemma symbolic run}, 
	$\Var(G_i)  =  \Var(g_i[\iota_i]) \subseteq \range(\iota_i) \subseteq \{ v_1 ,\ldots, v_i \}$, for $i \in \{ 1,\ldots, n\}$.
	By definition of $\delta$, $\guard(w) = G_1 \wedge \cdots \wedge G_n$ is satisfiable. \qed
\end{proof}
\fi

Since register automata are deterministic, each symbolic trace of $\A$ corresponds to a unique symbolic run of $\A$.

\begin{lemma}
\label{lem for each tainted trace there is a unique symbolic run}
Suppose $\delta$ and $\delta'$ are symbolic runs of a register automaton $\A$ such that $\strace(\delta) = \strace(\delta')$.
Then $\delta = \delta'$.
\end{lemma}
\iflong
\begin{proof}
We prove the lemma by contradiction. Suppose $\delta \neq \delta'$.
All symbolic runs of $\A$ share at least the initial configuration $(q_0, \zeta_0)$.
Let $\delta$ be as in Definition~\ref{def symbolic semantics}, and let $(q_{i-1}, \zeta_{i-1})$ be the last point where $\delta$ and $\delta'$ coincide.
From this point, $\delta$ continues its course with a step
\[
(q_{i-1}, \zeta_{i-1}) ~ \xrightarrow{\alpha_i, g_i, \varrho_i} ~ (q_i, \zeta_i),
\]
whereas $\delta'$ continues with a different step
\[
(q_{i-1}, \zeta_{i-1}) ~ \xrightarrow{\alpha'_i, g'_i, \varrho'_i} ~ (q'_i, \zeta'_i).
\]
Since $\strace(\delta) = \strace(\delta')$, $\alpha_i = \alpha'_i$ and $g_i[\iota_i] \equiv g'_i[\iota_i]$, where
$\iota_i = \zeta_{i-1} \cup \{ (p, v_i) \}$.
By Lemma~\ref{lemma injective valuations}, variable renaming $\iota_i$ is injective, which implies that $g_i \equiv g'_i$.
The underlying transitions $q_{i-1} \xrightarrow{\alpha_i, g_i, \varrho_i} q_i$ and $q_{i-1} \xrightarrow{\alpha'_i, g'_i, \varrho'_i} q'_i$ of $\A$ must be different, because otherwise also $\zeta_i$ and $\zeta'_i$ would be equal.
Therefore, since $\A$ is deterministic, $g_i \wedge g'_i$ is not satisfiable.  Since $g_i \equiv g'_i$, this means that $g_i$ is not satisfiable.
But since $\delta$ is a symbolic run, $G_i = g_i[\iota_i]$ is satisfiable, and thus there exists a valuation $\xi$ such that $\xi \models G_i$.  But now Lemma~\ref{lemma variable renaming} gives $\xi \circ \iota_i \models g_i$.  This means that $g_i$ is satisfiable and we have derived a contradiction. \qed
\end{proof}
\fi

Lemma~\ref{lem for each tainted trace there is a unique symbolic run} allows us to associate a unique symbolic run to each symbolic word that is accepted by a register automaton.

\begin{definition}
Let $\A$ be a register automaton and $w \in L_s(\A)$. Then we write $\symbolic(w)$ for the unique symbolic run $\delta$ of $\A$ with $\strace(\delta) = w$.
\end{definition}

There exists a one-to-one correspondence between runs of $\A$ and pairs consisting of a symbolic run of $\A$ and a satisfying assignments for the guards from its symbolic trace.

\begin{lemma}
	\label{lem run is a run}
	Let $\delta$ be a symbolic run of $\A$, as in Def.~\ref{def symbolic semantics}, and $\xi : \{ v_1 ,\ldots, v_n \} \rightarrow \D$ a valuation such that $\xi \models G_1 \wedge \cdots \wedge G_n$. Let $\run_{\A}(\delta, \xi)$ be the sequence obtained from $\delta$ by (a) replacing each input $\alpha_i$ by data symbol $\alpha_i(\xi(v_i))$ (for $0 < i \leq n$), (b) removing guards $g_i$ and assignments $\varrho_i$, and (c) replacing valuations $\zeta_i$ by $\xi_i = \xi \circ \zeta_i$ (for $0 \leq i \leq n$).
	Then $\run_{\A}(\delta,\xi)$ is a run of $\A$.
\end{lemma}
\iflong
\begin{proof}
It suffices to show, for $0 < i \leq n$, that
$\kappa_i \models g_i$, where $\kappa_i = \xi_{i-1} \cup \{ (p, d_i) \}$, and
$\xi_i = \kappa_i \circ \varrho_i$, for $0 < i \leq n$.
We derive
\[
\xi \circ \iota_i = \xi \circ (\zeta_{i-1} \cup \{ p, v_i \}) = \xi \circ \zeta_{i-1} \cup \{ p, d_i \} = \xi_{i-1} \cup \{ (p,d_i) \} = \kappa_i.
\]
By assumption, $\xi \models G_i \equiv g_i[\iota_i]$. By Lemma~\ref{lemma variable renaming}, $\xi \circ \iota_i \models g_i$. Hence, by the above derivation, $\kappa_i \models g_i$, as required. We derive
\[
\xi_i = \xi \circ \zeta_i = \xi \circ (\iota_i \circ \varrho_i) = (\xi \circ \iota_i) \circ \varrho_i = \kappa_i \circ \varrho_i.
\]
Thus $\xi_i = \kappa_i \circ \varrho_i$, as required. \qed
\end{proof}
\fi

\begin{lemma}
	\label{lem for each run unique symbolic run}
	Let $\gamma$ be a run of register automaton $\A$. Then there exist a valuation $\xi$ and symbolic run $\delta$ such that $\run_{\A}(\delta,\xi) = \gamma$.
\end{lemma}
\iflong
\begin{proof}
	Let $\gamma$ be as in Definition~\ref{def semantics register automata}:
	\begin{eqnarray*}
		\gamma & = & (q_0, \xi_0) ~ \xrightarrow{\alpha_1(d_1)} ~ 
		(q_1, \xi_1) \quad
		\ldots \quad
		(q_{n-1}, \xi_{n-1}) ~ \xrightarrow{\alpha_n(d_n)} ~ (q_n, \xi_n ),
	\end{eqnarray*}
	where, for $0 \leq i \leq n$, $(q_i, \xi_i)$ is a configuration of $\A$, $\domain(\xi_0) = \emptyset$, and for $0 < i \leq n$, $\Gamma$ contains a transition $q_{i-1} \xrightarrow{\alpha_i, g_i, \varrho_i} q_i$ such that
	\begin{itemize}
		\item
		$\iota'_i \models g_i$, where $\iota'_i = \xi_{i-1} \cup \{ (p, d_i) \}$, and
		\item
		$\xi_i = \iota'_i \circ \varrho_i$.
	\end{itemize}
Since $\A$ is deterministic, the transitions $q_{i-1} \xrightarrow{\alpha_i, g_i, \varrho_i} q_i$ are uniquely determined.
Let $\zeta_0$ be the trivial variable renaming with empty domain and, for $0 < i \leq n$, define $\zeta_i$ inductively by
$\zeta_i = \iota_i \circ \varrho_i$ and $\iota_i = \zeta_{i-1} \cup \{ (p, v_i) \}$.
Let $\xi : \{ v_1 ,\ldots, v_n \} \rightarrow \D$ be given by $\xi(v_i) = d_i$, for $1 \leq i \leq n$, and
let $\delta$ be the sequence
\[
\delta ~=~ (q_0, \zeta_0) ~ \xrightarrow{\alpha_1, g_1, \varrho_1} ~ 
(q_1, \zeta_1) ~
\ldots ~ \xrightarrow{\alpha_n, g_n, \varrho_n} ~
(q_n, \zeta_n).
\]
We claim that $\delta$ is a symbolic execution of $\A$.  For this, it suffices to show that $\xi \models G_1 \wedge \cdots \wedge G_n$, where $G_i \equiv g_i[\iota_i]$.
By induction on $i$ we show
\begin{eqnarray*}
	\iota'_i & = & \xi \circ \iota_i \mbox{ for } 0 < i \leq n\\
		\xi_i & = & \xi \circ \zeta_i \mbox{ for } 0 \leq i \leq n
\end{eqnarray*}
\begin{enumerate}
	\item 
	Base. $\xi_0 = \xi \circ \zeta_0$, as both $\xi_0$ and $\zeta_0$ have empty domain.
	\item
	Induction step.
	\begin{eqnarray*}
		\iota'_i & = & \xi_{i-1} \cup \{ (p, d_i) \} = \mbox{(by induction hypothesis)}\\
		&= & \xi \circ \zeta_{i-1} \cup \{ (p, d_i) \} = \xi \circ (\zeta_{i-1} \cup \{(p, v_i) \}) = \xi \circ \iota_i\\
		\xi_i & = & \iota'_i \circ \varrho_i = (\xi \circ \iota_i) \circ \varrho_i = \xi \circ (\iota_i \circ \varrho_i) = \xi \circ \zeta_i
	\end{eqnarray*}
\end{enumerate}
Let $0 \leq i \leq n$. Since $\gamma$ is a run, $\iota'_i \models g_i$.  By the identity we  just derived,
$\xi \circ \iota_i \models g_i$. By Lemma~\ref{lemma variable renaming}, $\xi \models g_i[\iota_i] \equiv G_i$.
Hence $\xi \models G_1 \wedge \cdots \wedge G_n$, which proves the claim that $\delta$ is a symbolic execution of $\A$. 
It is easy to verify that $\gamma = \run(\delta,\xi)$. \qed
\end{proof}
\fi

Using the above lemmas, we can show that
whenever two register automata accept the same symbolic language, they also accept the same data language.

\begin{theorem}
	\label{theorem symbolic language equivalence refines language equivalence}
Suppose $\A$ and $\B$ are register automata with $L_s(\A) = L_s(\B)$. Then $L(\A) = L(\B)$.
\end{theorem}
\iflong
\begin{proof}
	We will prove $L(\A) \subseteq L(\B)$.  The proof of the inclusion $L(\A) \subseteq L(\B)$ is symmetric.
	Suppose $w \in L(\A)$.
	Then there exists a run $\gamma$ of $\A$ with $\trace(\gamma) = w$.
	By Lemma~\ref{lem for each run unique symbolic run}, there exist a valuation $\xi$ and symbolic run $\delta$ of $\A$ such that $\run_{\A}(\delta,\xi) = \gamma$.
	Let $u = \strace(\delta)$.
	Then $u \in L_s(\A)$ and, since $L_s(\A) = L_s(\B)$, $u \in L_s(\B)$.
	Let $\delta'$ be a symbolic run of $\B$ such that $\strace(\delta') = u$.
	Let $\gamma' = \run_{\B}(\delta', \xi)$. By Lemma~\ref{lem run is a run}, $\gamma'$ is a run of $\B$. Let $w' = \trace(\gamma')$. Then $w' \in L(\B)$.
	Note that $w$ and $w'$ share the same sequence of data values, as given by valuation $\xi$.
	Also note that $w$, $\gamma$, $\delta$, $u$, $\delta'$, $\gamma'$ and $w'$ all share the same sequence of input symbols.
	Thus $w = w'$ and $w \in L(\B)$, as required. \qed
\end{proof}
\fi

\begin{example}
The converse of Theorem~\ref{theorem symbolic language equivalence refines language equivalence} does not hold.
Figure~\ref{fig counterexample} gives a trivial example of two register automata with the same data language but a different symbolic language.
\begin{figure}[h!]
	\begin{center}
		\begin{tikzpicture}[shorten >=1pt,node distance=2cm,on grid,auto] 
		\node[state,initial,accepting] (q_0)   {$q_0$}; 
			\node[state,color=white] (q_x) [ right=of q_0] {};
		\node[state,initial, accepting] (q_1) [ right=of q_x] {$q_1$}; 
		\node[state,accepting] (q_2) [below=of q_0] {$q_2$};
		\node[state,accepting] (q_3) [below=of q_1] {$q_3$};
		\path[->]
		(q_0) edge [bend left] node[right] {$a,  p > 0$} (q_2)
		(q_0) edge [bend right] node[left] {$a,  p\leq 0$} (q_2)
		(q_1) edge  node {$a$} (q_3);
		\end{tikzpicture}
		\caption{Two register automata that are trace equivalent but not symbolic trace equivalent.}
		\label{fig counterexample}
	\end{center}
\end{figure}
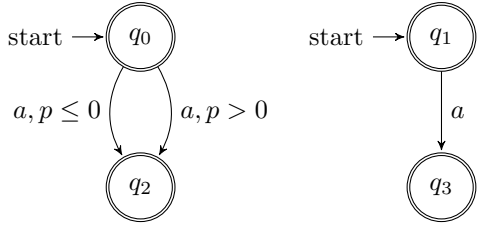
\end{example}

Lemma~\ref{lem run is a run} allows us to rephrase the well-formedness condition of register automata in terms of symbolic runs.

\begin{corollary}
	\label{variables guards always defined V2}	
Register automaton $\A$ is well-formed iff, for each symbolic run $\delta$ that ends with $(q, \zeta)$,
$q \xrightarrow{\alpha, g \varrho} q'  \Rightarrow  \Var(g) \subseteq \domain(\zeta) \cup \{ p \}$.
\end{corollary}
\iflong
\begin{proof}
	``$\Rightarrow$''
Suppose symbolic execution $\delta$ is defined as in Definition~\ref{def symbolic semantics}. Let $\xi : \{ v_1 ,\ldots, v_n \} \rightarrow \D$ be a valuation such that $\xi \models G_1 \wedge \cdots \wedge G_n$. (Such a valuation $\xi$ exists since, by definition of a symbolic run, $G_1 \wedge \cdots \wedge G_n$ is satisfiable.) By Lemma~\ref{lem run is a run}, $\gamma = \run(\delta,\xi)$ is a run of $\A$.  By construction of $\gamma$, $\gamma$ ends with a reachable configuration $(q, \xi)$, where $\domain(\xi) = \domain(\zeta)$.  Now we may apply the definition of well-formedness to conclude $\Var(g) \subseteq \domain(\zeta) \cup \{ p \}$.

``$\Leftarrow$''
Suppose that for each symbolic run $\delta$
that ends with $(q, \zeta)$, we have
$q \xrightarrow{\alpha, g \varrho} q'  \Rightarrow  \Var(g) \subseteq \domain(\zeta) \cup \{ p \}$.
Let $(q, \xi)$ be the final configuration of a run $\gamma$ of $\A$ with $q \xrightarrow{\alpha, g \varrho} q'$.
By Lemma~\ref{lem for each run unique symbolic run},
there exist a valuation $\xi'$ and symbolic run $\delta$ such that $\run_{\A}(\delta,\xi') = \gamma$.
Let $(q, \zeta)$ be the final configuration of symbolic run $\delta$. 
Then by the assumption, $\Var(g) \subseteq \domain(\zeta) \cup \{ p \}$.
Therefore, since $\domain(\xi)=\domain(\zeta)$, $\Var(g) \subseteq \domain(\zeta) \cup \{ p \}$ and we may conclude that $\A$ is well-formed. \qed
\end{proof}
\fi

\section{A Myhill-Nerode Theorem}
The Nerode equivalence \cite{nerode58,HU79} deems two words $w$ and $w'$ of a language $L$ equivalent if there does not exist a suffix $u$ that distinguishes them, that is, only one of the words $w \cdot u$ and $w' \cdot u$ is in $L$. The Myhill-Nerode theorem states that $L$ is regular if and only if this equivalence relation has a finite index, and moreover that the number of states in the smallest deterministic finite automaton (DFA) recognizing $L$ is equal to the number of equivalence classes.
In this section, we present a Myhill-Nerode theorem for symbolic languages and register automata. We use three relations $\equiv_l$, $\equiv_t$ and $\equiv_r$ on symbolic words to capture the structure of register automata. Intuitively,
symbolic words $w$ and $w'$ are \emph{location equivalent}, notation $w \equiv_l w'$, if they lead to the same location, \emph{transition equivalent}, notation $w \equiv_t w'$, if they share the same final transition, and marker $v$ of $w$, and marker $v'$ of $w'$ are \emph{register equivalent}, notation $(w,v) \equiv_r (w',v')$, when they are stored in the same register after occurrence of words $w$ and $w'$.
Whereas $\equiv_l$ and $\equiv_t$ are equivalence relations, $\equiv_r$ is a partial equivalence relation (PER), that is, a relation that is symmetric and transitive.  Relation $\equiv_r$ is not necessarily reflexive, as $(w,v) \equiv_r (w,v)$ only holds when marker $v$ is stored after symbolic trace $w$.
Since a register automaton has finitely many locations, finitely many transitions, and finitely many registers,
the equivalences $\equiv_l$ and $\equiv_t$, and the equivalence induced by $\equiv_r$, are all required to have finite index.

\begin{definition}
\label{def regular}
A feasible symbolic language $L$ over $\Sigma$ is \emph{regular} iff there exist three relations:
\begin{itemize}
\item
an equivalence relation $\equiv_l$ on $L$, called \emph{location equivalence},
\item
an equivalence relation $\equiv_t$ on $L \setminus \{ \epsilon\}$, called \emph{transition equivalence}, and
\item
a partial equivalence relation $\equiv_r$ on $\{ (w,v_i) \in L \times \V \mid i \leq \length(w) \}$, called \emph{register equivalence}, satisfying
$(w,v) \equiv_r (w,v')  \Rightarrow  v=v'$.
We say that $w$ \emph{stores} $v$ if $(w, v) \equiv_r (w, v)$.
\end{itemize}
We require that equivalences $\equiv_l$ and $\equiv_t$, as well as the equivalence relation obtained by restricting $\equiv_r$ to $\{ (w, v) \in L \times \V \mid w \mbox{ stores } v \}$ have finite index.
Given $w$, $w'$ and $v$, there is at most one $v'$ s.t.\  $(w,v) \equiv_r (w',v')$. Therefore, we may define $\matching(w, w')$ as the variable renaming $\sigma$ satisfying:
\begin{eqnarray*}
\sigma(v) & = & \left\{ \begin{array}{ll} v'  & \mbox{if } (w,v) \equiv_r (w',v')\\
	v_{n+1} & \mbox{if } v=v_{m+1}\\
	\mbox{undefined} & \mbox{otherwise}
\end{array} \right.
\end{eqnarray*}	
Finally, we require that relations $\equiv_l$, $\equiv_t$ and $\equiv_r$ satisfy the conditions of Table~\ref{table conditions}, for
$w, w', u, u' \in L$, $\length(w)=m$, $\length(w')=n$, $\alpha, \alpha' \in \Sigma$, $G, G'$ guards, $v, v' \in \V$, and  $\sigma : \V \rightharpoonup \V$. 
\begin{table}
\begin{tcolorbox}	
\begin{eqnarray}
&& (w,v) \equiv_r (w,v')  \Rightarrow  v=v' \label{cond r marker}  \\
&&w \alpha G \equiv_t w' \alpha' G'  ~ \Rightarrow ~ w \equiv_l w'  
\label{cond t source}\\
&&w \alpha G \equiv_t w' \alpha' G'  ~ \Rightarrow ~  \alpha=\alpha'  
\label{cond t input}\\
&&w \alpha G \equiv_t w' \alpha G'  \wedge \sigma=\matching(w, w') ~ \Rightarrow ~ G[\sigma] \equiv G'  
\label{cond t guard}\\
&&w \equiv_t w'  \Rightarrow  w \equiv_l w'  
\label{cond t target}\\
&& w \equiv_t w' \wedge w \mbox{ stores } v_m  \Rightarrow (w, v_m) \equiv_r (w', v_n)  
\label{cond t parameter}\\
&& u \equiv_t  u'  \wedge u = w \alpha G \wedge u' = w' \alpha G' \wedge  (w,v) \equiv_r (w',v') \wedge u \mbox{ stores } v\nonumber\\
&& \quad\quad \Rightarrow (u, v) \equiv_r (u', v')  
\label{cond r forward propagation}\\
&& u \equiv_t  u'  \wedge u = w \alpha G \wedge u' = w' \alpha G' \wedge  (u,v) \equiv_r (u',v') \wedge v \neq v_{m+1} \nonumber\\
&& \quad\quad \Rightarrow  (w, v) \equiv_r (w', v') 
\label{cond r backward propagation}\\
&& w \equiv_l w' \wedge w \alpha G \in L \wedge v \in \Var(G)\setminus\{ v_{m+1} \} \nonumber\\
&& \quad\quad \Rightarrow  \exists v' : (w,v) \equiv_r (w',v') 
\label{cond valuation defined for variables in guards}	\\
&& w \equiv_l w'  \wedge w \alpha G \in L \wedge \sigma=\matching(w, w')  \nonumber\\
&& \quad\quad  \wedge  ~ \satisfiable(\guard(w') \wedge  G[\sigma]) \Rightarrow  w' \alpha G[\sigma] \in L  
\label{cond right invariance}\\
&& w \equiv_l w' \wedge w \alpha G \in L \wedge w' \alpha G' \in L \wedge \sigma=\matching(w, w')\nonumber \\
&& \quad\quad \wedge ~ \satisfiable(G[\sigma] \wedge G')  \Rightarrow  w \alpha G \equiv_t w' \alpha G'   
\label{cond determinism}
\end{eqnarray}
\end{tcolorbox}
\caption{Conditions for regularity of symbolic languages.}
\label{table conditions}
\end{table}
\end{definition}
Intuitively, the first condition captures that a register can store at most a single value at a time.
When $w \alpha G$ and $w' \alpha' G'$ share the same final transition, then in particular $w$ and $w'$ share the same final location (Condition~\ref{cond t source}), input symbols $\alpha$ and $\alpha'$ are equal (Condition~\ref{cond t input}), $G'$ is just a renaming of $G$ (Condition~\ref{cond t guard}), and $w \alpha G$ and $w' \alpha' G'$ share the same final location (Condition~\ref{cond t target}) and final assignment
(Conditions \ref{cond t parameter}, \ref{cond r forward propagation} and \ref{cond r backward propagation}).
Condition~\ref{cond t parameter} says that the parameters of the final input end up in the same register when they are stored.
Condition~\ref{cond r forward propagation} says that when two values are stored in the same register, they will stay in the same register as long as they are stored (this condition can be viewed as a right invariance condition for registers).
Conversely, if two values are stored in the same register after a transition, and they do not correspond to the final input, they were already  stored in the same register before the transition (Condition~\ref{cond r backward propagation}).
Condition~\ref{cond valuation defined for variables in guards} captures the well-formedness assumption for register automata.
As a consequence of Condition~\ref{cond valuation defined for variables in guards}, $G[\sigma]$ is defined in Conditions~\ref{cond t guard}, \ref{cond right invariance} and \ref{cond determinism}, since $\Var(G) \subseteq\domain(\sigma)$.
Condition~\ref{cond right invariance} is the equivalent for symbolic languages of the well-known right invariance condition for regular languages.
For symbolic languages a right invariance condition 
\begin{eqnarray*}
	&& w \equiv_l w'  \wedge w \alpha G \in L \wedge \sigma=\matching(w, w')  \Rightarrow  w' \alpha G[\sigma] \in L
\end{eqnarray*}
would be too strong:
even though $w$ and $w'$ lead to the same location, the values stored in the registers may be different, and therefore they will not necessarily enable the same transitions.
However, when in addition $\guard(w') \wedge G[\sigma]$ is satisfiable, we may conclude that $w'' \alpha G[\sigma] \in L$.	
Condition~\ref{cond determinism}, finally, asserts that $L$ only allows deterministic behavior.

The simple lemma below asserts that, due to the determinism imposed by Condition~\ref{cond determinism}, the converse of Conditions~\ref{cond t source}, \ref{cond t input} and \ref{cond t guard} combined also holds. This means that $\equiv_t$ can be expressed in terms of $\equiv_l$ and $\equiv_r$, that is, once we have fixed $\equiv_l$ and $\equiv_r$, relation $\equiv_t$ is fully determined.

\begin{lemma}
\label{lem determinism corollary}
Suppose symbolic language $L$ over $\Sigma$ is regular, and equivalences $\equiv_l$, $\equiv_t$ and $\equiv_r$ satisfy the conditions of Definition~\ref{def regular}. Then
\begin{eqnarray*}
&& w \equiv_l w' \wedge w \alpha G \in L \wedge w' \alpha G' \in L \wedge \sigma=\matching(w, w') \wedge G' \equiv G[\sigma]\nonumber\\
&&\quad\quad \Rightarrow w \alpha G \equiv_t w' \alpha G'.
\end{eqnarray*}
\end{lemma}
\iflong
\begin{proof}
Suppose the left hand side of the above implication holds.
Since $L$ is regular, it is in particular feasible, and therefore $G'$ is satisfiable. But then, since $G' \equiv G[\sigma]$, also $G[\sigma] \wedge G'$ is satisfiable. Therefore, Condition~\ref{cond determinism} implies that the right hand side of the implication holds. \qed
\end{proof}
\fi

\begin{example}
	Even though $\equiv_t$ can be expressed in terms of $\equiv_l$ and $\equiv_r$, there are symbolic languages that satisfy all the conditions for regularity, except the condition that $\equiv_t$ has finite index. So when $\equiv_l$ and $\equiv_r$ have finite index, this does not imply that $\equiv_t$ has finite index. An example of such a language is:
	\begin{eqnarray*}
		L & = & \{ \epsilon, a ~ v_1 = 1, a ~ v_1 =2, a ~ v_1 = 3, \ldots \}.
	\end{eqnarray*}
Here we assume the set $R$ of relation symbols contains unary relations ``$. = i$'', for each natural number $i$.
For language $L$ we may define an equivalence relation $\equiv_l$ comprising two equivalence classes $\{ \epsilon \}$ and $L \setminus \{ \epsilon \}$.
No values need to be stored and we may thus define $\equiv_r$ to be the empty relation.
The guards of all nonempty symbolic words are different, both syntactically and semantically.
Therefore, by Condition~\ref{cond t  guard}, $\equiv_t$ must have infinitely many equivalence classes, one for each nonempty word in $L$.
The reader may check that, with these definitions of $\equiv_l$, $\equiv_t$ and $\equiv_r$, all conditions of regularity are met, except that $\equiv_t$ has infinite index.
\end{example}	
	
We can now state 
\iflong 
and prove
\fi
our ``symbolic'' version of the celebrated result of Myhill \& Nerode.
\iflong
First we prove that the 
\else
The
\fi
symbolic language of any register automaton is regular (Theorem~\ref{theorem soundness}), and
\iflong
then we establish that 
\fi
any regular symbolic language can be obtained as the symbolic language of some register automaton (Theorem~\ref{theorem completeness}).

\begin{theorem}
\label{theorem soundness}
Suppose $\A$ is a register automaton.  Then $L_s(\A)$ is regular.
\end{theorem}
\begin{proof}
	\iflong
Let $L = L_s(\A)$.
\else
\emph{(outline)} Let $L = L_s(\A)$.
\fi
Then, by Lemma~\ref{lemma tainted language register automaton is feasible}, $L$ is feasible.
Define equivalences $\equiv_l$, $\equiv_t$ and $\equiv_r$ as follows:
\begin{itemize}
	\item
	For $w, w' \in L$, $w \equiv_l w'$ iff $\btr(w)$ and $\btr(w')$ share the same final location. 
	\item
	For $w, w' \in L \setminus \{ \epsilon\}$, $w \equiv_t w'$ iff $\btr(w)$ and $\btr(w')$ share the same final transition. 	
	\item
	For $w, w' \in L$ and $v, v' \in \V$,  $(w,v) \equiv_r (w',v')$ iff there is a register $x \in V$ such that the final valuations $\zeta$ of $\btr(w)$ stores $v$ in $x$, and the final valuation $\zeta'$ of $\btr(w')$ stores $v'$ in $x$, that is, $\zeta(x) = v$ and $\zeta'(x) = v'$.\\
	(Note that, by Lemma~\ref{lemma symbolic run}, $\range(\zeta) \subseteq \{ v_1,\ldots, v_m \}$, for $m = \length(w)$, and $\range(\zeta') \subseteq \{ v_1,\ldots, v_n \}$, for $n = \length(w')$.) 
\end{itemize} 
Then $\equiv_l$ has finite index since $\A$ has a finite number of locations,
$\equiv_t$ has finite index since $\A$ has a finite number of transitions, and
the equivalence induced by $\equiv_r$ has finite index since $\A$ has a finite number of registers.
\iflong

Assume $w, w' \in L$, where $w$ contains $m$ input symbols and $w'$ contains $n$ input symbols.
%$v, v' \in \V$ and $\sigma$ is a variable renaming.
Let
\begin{eqnarray*}
\btr(w) & = & (q_0, \zeta_0) ~ \xrightarrow{\alpha_1, g_1, \varrho_1} ~ 
(q_1, \zeta_1) ~
\ldots ~ (q_{m-1}, \zeta_{m-1}) \xrightarrow{\alpha_m, g_m, \varrho_m} ~
(q_m, \zeta_m),\\
\btr(w') & = & (q'_0, \zeta'_0) ~ \xrightarrow{\alpha'_1, g'_1, \varrho'_1} ~ 
(q'_1, \zeta'_1) ~
\ldots ~ (q'_{n-1}, \zeta'_{n-1}) \xrightarrow{\alpha'_n, g'_n, \varrho'_n} ~
(q'_n, \zeta'_n),
\end{eqnarray*}
as in Definition~\ref{def symbolic semantics}.
%Then $\alpha = \alpha_n$, $\alpha' = \alpha'_n$, $G = g_n[\iota]$ and $G' = g'_m[\iota']$,
%where $\iota = \zeta_{n-1} \cup \{ (p,v_n) \}$ and $\iota'= \zeta'_{m-1} \cup \{ (p, v_m) \}$.
%
We show that all 11 conditions of Table~\ref{table conditions} hold:
\begin{itemize}
\item 
Condition~\ref{cond r marker}.  If $(w,v) \equiv_r (w,v')$ then $v$ and $v'$ are stored in the same register $x$ in the final valuation $\zeta_m$ of $\btr(w)$. Thus $v = \zeta_m(x) = v'$. 	
\item 
Condition~\ref{cond t source}. If $\btr(w \alpha G)$ and $\btr(w' \alpha' G')$ share the same final transition, then $\btr(w)$ and $\btr(w')$ certainly share the same final location.
\item 
Condition~\ref{cond t input}. If $\btr(w \alpha G)$ and $\btr(w' \alpha' G')$ share the same final transition, then $\alpha$ and $\alpha'$ must be equal to the input symbols of this final transition, and thus equal to each other.
\item 
Condition~\ref{cond t guard}. Assume $w \alpha G \equiv_t w' \alpha G'$ and $\sigma=\matching(w,w')$.
Let $\btr(w \alpha G)$ and $\btr(w' \alpha G')$ be obtained by appending transitions
\[
(q_m, \zeta_m)  \xrightarrow{\alpha, g, \varrho} ~ (q, \zeta) \mbox{ and }
(q'_n, \zeta'_n)  \xrightarrow{\alpha, g', \varrho'} ~ (q', \zeta')
\]
to $\btr(w)$ and $\btr(w')$, respectively.
Then $q_m = q'_n$, $g \equiv g'$, $\varrho = \varrho'$, $q=q'$,
$G \equiv g[\iota]$, where $\iota = \zeta_m \cup \{ (p,v_{m+1}) \}$, and
$G' \equiv g'[\iota']$, where $\iota' = \zeta'_n \cup \{ (p,v_{n+1}) \}$.
We have to show that $G' \equiv G[\sigma]$, or equivalently $g[\sigma\circ\iota] = g[\iota']$.
Suppose $x \in \Var(g)$.
\begin{itemize}
	\item 
	If $x =p$ then $\sigma\circ\iota(x) = \sigma\circ\iota(p) = \sigma(v_{m+1}) = v_{n+1} = \iota'(p) = \iota'(x)$.
	\item 
	If $x \neq p$ then,
	by Corollary~\ref{variables guards always defined V2}, $x \in \domain(\zeta_m)$ and $x \in \domain(\zeta'_n)$.
	Let $v = \zeta_m(x)$ and $v' = \zeta'_n(x)$.  Then, by definition of $\equiv_r$, $(w, v) \equiv_r (w',v')$ and thus $\sigma(v) = v'$.
	Hence $\sigma \circ \iota (x) = \sigma \circ \zeta_m (x) = \sigma(v) = v' = \zeta'_m(x) = \iota'(x)$.
\end{itemize}
\item 
Condition~\ref{cond t target}. If $\btr(w)$ and $\btr(w')$ share the same final transition, they certainly share the same final location.
\item 
Condition~\ref{cond t parameter}. 
Assume $w \equiv_t w'$ and $w$ stores $v_m$.
Then there exists a variable $x$ such that $\zeta_m(x) = v_m$.
By the definition of symbolic runs,
$\zeta_m = \iota_m \circ \varrho_m$, where $\iota_m = \zeta_{m-1} \cup \{ (p, v_m) \}$.
By Lemma~\ref{lemma symbolic run}, $\range(\zeta_{m-1}) \subseteq \{ v_1,\ldots, v_{m-1} \}$.
We conclude that $\varrho_m(x) = p$.
Again by the definition of symbolic runs,
$\zeta'_n = \iota'_n \circ \varrho'_n$, where $\iota'_n = \zeta'_{n-1} \cup \{ (p, v_n) \}$.
Since $w \equiv_t w'$, we know $\varrho_m = \varrho'_n$.
Therefore $\zeta'_n (x) =  \iota'_n \circ \varrho'_n (x) =  \iota'_n \circ \varrho_m(x) = \iota'_n (p)= v_n$.
This implies $(w, v_m) \equiv_r (w', v_n)$, as required.
\item 
Condition~\ref{cond r forward propagation}. 
Assume that $u \equiv_t  u'$, $u = w \alpha G$, $u' = w' \alpha G'$, $u$ stores $v$, and $(w,v) \equiv_r (w',v')$.
Let $\btr(w \alpha G)$ and $\btr(w' \alpha G')$ be obtained by appending transitions
\[
(q_m, \zeta_m)  \xrightarrow{\alpha, g, \varrho} ~ (q, \zeta) \mbox{ and }
(q'_n, \zeta'_n)  \xrightarrow{\alpha, g', \varrho'} ~ (q', \zeta')
\]
to $\btr(w)$ and $\btr(w')$, respectively.
Then $\varrho = \varrho'$, 
$\zeta = \iota \circ \varrho$, where $\iota = \zeta_m \cup \{ (p,v_{m+1}) \}$, and
$\zeta' = \iota' \circ \varrho$, where $\iota' = \zeta'_n \cup \{ (p,v_{n+1}) \}$.
Since $(w,v) \equiv_r (w',v')$, there exists an $x \in V$ such that $\zeta_m(x) = v$ and $\zeta'_n(x) = v'$.
Thus also $\iota(x) = v$ and $\iota'(x) = v'$.
Since $u$ stores $v$, there exists an $y \in V$ such that $\zeta(y) = v$.
By Lemma~\ref{lemma injective valuations}, $\iota$ is injective.
Thus $\iota(\varrho(y)) = v$ and $\iota(x) = v$ implies $\varrho(y) = x$.
But this means $\zeta'(y) = \iota' \circ \varrho (y) = \iota'(x) = v'$.
Therefore $(u, v) \equiv_r (u', v')$.
\item
Condition~\ref{cond r backward propagation}. 
Assume $u \equiv_t  u'$, $u = w \alpha G$, $u' = w' \alpha G'$, $v \neq v_{m+1}$ and $(u,v) \equiv_r (u',v')$.
Let $\btr(w \alpha G)$ and $\btr(w' \alpha G')$ be obtained by appending transitions
\[
(q_m, \zeta_m)  \xrightarrow{\alpha, g, \varrho} ~ (q, \zeta) \mbox{ and }
(q'_n, \zeta'_n)  \xrightarrow{\alpha, g', \varrho'} ~ (q', \zeta')
\]
to $\btr(w)$ and $\btr(w')$, respectively.
Then $\varrho = \varrho'$, 
$\zeta = \iota \circ \varrho$, where $\iota = \zeta_m \cup \{ (p,v_{m+1}) \}$, and
$\zeta' = \iota' \circ \varrho$, where $\iota' = \zeta'_n \cup \{ (p,v_{n+1}) \}$.
Since $(u,v) \equiv_r (u',v')$, there exists an $x \in V$ such that $\zeta(x) = v$ and $\zeta'(x) = v'$.
Using $v \neq v_{m+1}$, we infer that there exists an $y \in V$ such that $\varrho(x) = y$ and $\zeta_m(y) = v$.
Now we derive $\zeta'_n(y) = \iota'(y) = \iota' \circ \varrho(x) = \zeta'(x) = v'$.
Therefore $(w, v) \equiv_r (w', v')$.
\item 
Condition~\ref{cond valuation defined for variables in guards}.
Assume $w \equiv_l w'$, $w \alpha G \in L$ and $v \in \Var(G)\setminus\{ v_{m+1} \}$.
Let $\btr(w \alpha G)$ be obtained by appending transition
\[
(q_m, \zeta_m)  \xrightarrow{\alpha, g, \varrho} ~ (q, \zeta)
\]
to $\btr(w)$.
Then $q_m = q'_n$ and $G \equiv g[\iota]$, where $\iota = \zeta_m \cup \{ (p,v_{m+1}) \}$.
%By Lemma~\ref{lemma symbolic run}, $\range(\zeta_m) \subseteq \{ v_1,\ldots, v_m \}$.
Since $v \in \Var(G)\setminus\{ v_{m+1} \}$, there exists a variable $x \in \Var(g) \setminus \{ p \}$ with $\zeta_m(x) =v$.
By Corollary~\ref{variables guards always defined V2}, $\Var(g) \subseteq \domain(\zeta'_n) \cup \{ p \}$, and
thus $x \in \domain(\zeta'_n)$.  Let $v' = \zeta'_n(x)$.
Then $(w,v) \equiv_r (w',v')$.
\item 
Condition~\ref{cond right invariance}.
Assume that $w \equiv_l w'$, $w \alpha G \in L$, $\sigma=\matching(w, w')$ and $\satisfiable(\guard(w') \wedge  G[\sigma])$.
Since $w \alpha G \in L$, $\btr(w \alpha G)$ can be obtained by appending a transition
\[
(q_m, \zeta_m)  \xrightarrow{\alpha, g, \varrho} ~ (q, \zeta)
\]
to $\btr(w)$, with $G \equiv g[\iota]$, where $\iota = \zeta_m \cup \{ (p,v_{m+1}) \}$.
Since $w \equiv_l w'$, $q_m = q'_n$.
Now consider the sequence $\delta'$ obtained by appending a transition
\[
(q'_n, \zeta'_n)  \xrightarrow{\alpha, g, \varrho} ~ (q, \zeta')
\]
to $\btr(w')$, with $\zeta' = \iota' \circ \varrho$, where $\iota' = \zeta'_n \cup \{ (p,v_{n+1}) \}$.
Since $\guard(w') \wedge  G[\sigma]$ is satisfiable, we may conclude that $\delta'$ is a symbolic execution if we
can prove $G[\sigma] \equiv g[\iota']$, or equivalently $g[\sigma\circ\iota] = g[\iota']$.
Suppose $x \in \Var(g)$.
\begin{itemize}
	\item 
	If $x =p$ then $\sigma\circ\iota(x) = \sigma\circ\iota(p) = \sigma(v_{m+1}) = v_{n+1} = \iota'(p) = \iota'(x)$.
	\item 
	If $x \neq p$ then,
	by Corollary~\ref{variables guards always defined V2}, $x \in \domain(\zeta_m)$ and $x \in \domain(\zeta'_n)$.
	Let $v = \zeta_m(x)$ and $v' = \zeta'_n(x)$.  Then, by definition of $\equiv_r$, $(w, v) \equiv_r (w',v')$ and thus $\sigma(v) = v'$.
	Hence $\sigma \circ \iota (x) = \sigma \circ \zeta_m (x) = \sigma(v) = v' = \zeta'_m(x) = \iota'(x)$.
\end{itemize}
Hence $g[\sigma\circ\iota] = g[\iota']$ and $\delta'$ is a symbolic run for $w' \alpha G[\sigma]$.
We conclude $w' \alpha G[\sigma] \in L$.
\item 
Condition~\ref{cond determinism}. Suppose 
$w \equiv_l w'$, $ w \alpha G \in L$, $w' \alpha G' \in L$, $\sigma=\matching(w, w')$ and
$G[\sigma] \wedge G'$ is satisfiable.
Let $\delta = \btr(w \alpha G)$ and $\delta' = \btr(w' \alpha G')$ be obtained by appending transitions
\[
(q_m, \zeta_m)  \xrightarrow{\alpha, g, \varrho} ~ (q, \zeta) \mbox{ and }
(q'_n, \zeta'_n)  \xrightarrow{\alpha, g', \varrho'} ~ (q', \zeta')
\]
to $\btr(w)$ and $\btr(w')$, respectively.
Then 
$G \equiv g[\iota]$, where $\iota = \zeta_m \cup \{ (p,v_{m+1}) \}$, and
$G' \equiv g'[\iota']$, where $\iota' = \zeta'_n \cup \{ (p,v_{n+1}) \}$.
Since $G[\sigma] \wedge G'$ is satisfiable, there exists a valuation $\xi$ such that
\[
\xi \models G[\sigma] \wedge G'.
\]
Define variable renaming $\sigma'$ as follows
\begin{eqnarray*}
\sigma'(x) & = & \left\{ \begin{array}{ll} 
\iota'(x) & \mbox{if } x \in \Var(g')\\
\sigma \circ \iota(x) & \mbox{otherwise}
\end{array}\right.
\end{eqnarray*}
Then clearly $G' \equiv g'[\iota'] \equiv g'[\sigma']$.
We verify that $G[\sigma] \equiv g[\sigma\circ\iota]\equiv g[\sigma']$. Let $x \in \Var(g)$. Then
\begin{itemize}
	\item 
	If $x = p$ then $\sigma\circ\iota(x) = \sigma\circ\iota(p) = \sigma(v_{m+1}) = v_{n+1} = \iota'(p) = \iota'(x)$.
	\item 
	If $x \in\Var(g') \setminus \{ p \}$ then, by Corollary~\ref{variables guards always defined V2}, $x \in\domain(\zeta_m)$ and $x \in\domain(\zeta'_n)$.  Let $\zeta_m(x) = v$ and $\zeta'_n(x) = v'$. Then $(w,v) \equiv_r (w',v')$ and thus $\sigma(v) = v'$. Hence $\sigma\circ\iota(x) = \sigma\circ\zeta_m(x) = \sigma(v) = v' = \zeta'_n(x) = \iota'(x)$.
	\item 
	If $x \not\in\Var(g')$ then, by definition of $\sigma'$, $\sigma\circ\iota(x) = \sigma'(x)$.
\end{itemize}
Thus
\[
G[\sigma] \wedge G' \equiv (g \wedge g')[\sigma'].
\]
Therefore $\xi \models (g \wedge g')[\sigma']$ and, by Lemma~\ref{lemma variable renaming}, $\xi \circ \sigma' \models g \wedge g'$.
This means that $g \wedge g'$ is satisfiable.
Since $w \equiv_l w'$, $q_m = q'_n$.
Because $\A$ is required to be deterministic, the conjunction of the guards of any pair of distinct $\alpha$-transitions from
$q_m = q'_n$ is not satisfiable.  Therefore the final transitions of $\delta$ and $\delta'$ must be equal.
This implies $w \alpha G \equiv_t w' \alpha G'$. \qed
\end{itemize}
\else
We refer to the full version of this paper for a proof that, with this definition of $\equiv_l$, $\equiv_t$ and $\equiv_r$,
all 11 conditions of Table~\ref{table conditions} hold. \qed
\fi
\end{proof}

The following example shows that in general there is no coarsest location equivalence that satisfies all conditions of Table~\ref{table conditions}. So whereas for regular languages a unique Nerode equivalence exists, this is not always true for symbolic languages.

\begin{example}
	\label{example location equivalence not unique}
	Consider the symbolic language $L$ that consists of the following three symbolic words and their prefixes:
	\begin{eqnarray*}
		w  & ~~~=~~~ & a ~ v_1 > 0 ~ a ~ v_1 > 0 ~ b ~ \top\\	
		u &~~~ =~~~ & a ~ v_1 = 0 ~ a ~ v_1 = 0 ~ b ~ \top\\
		z &~~ =~~~ & a ~ v_1 < 0 ~ c ~ v_1 + v_2 = 0 ~ a ~ v_2 > 0 ~ c ~ \top
	\end{eqnarray*}
	Symbolic language $L$ is accepted by both automata displayed in Figure~\ref{fig: location equivalence not uniquely determined}. Thus, by Theorem~\ref{theorem soundness}, $L$ is regular.
	Let $w_i$, $u_i$ and $z_i$ denote the prefixes of $w$, $u$ and $z$, respectively, of length $i$.
	Then, according to the location equivalence induced by the first automaton, $w_1 \equiv_l u_1$, and
	according to the location equivalence induced by the second automaton, $u_1 \equiv_l z_2$.
	Therefore, if a coarsest location equivalence relation would exist, $w_1 \equiv_l z_2$ should hold.
	Then, by Condition~\ref{cond valuation defined for variables in guards},
	$(w_1,v_1) \equiv_r (z_2, v_2)$. Thus, by Lemma~\ref{lem determinism corollary}, $w_2 \equiv_t z_3$, and therefore,
	by Condition~\ref{cond t target}, $w_2 \equiv_l z_3$.
	But now Condition~\ref{cond right invariance} implies $a ~ v_1 > 0 ~ a ~ v_1 > 0 ~ c ~ \top \in L$, which is a contradiction.
	\begin{figure}[htb!]
		\begin{center}
			\begin{tikzpicture}[shorten >=1pt,node distance=3cm,on grid,auto] 
			\node[state,initial,accepting] (q_0)   {$q_0$}; 
			\node[state,accepting] (q_1) [ below=of q_0] {$q_1$};
			\node[state,accepting] (q_2) [ right=of q_0] {$q_2$}; 
			\node[state,accepting](q_3) [right=of q_1] {$q_3$};
			\node[state,accepting](q_4) [right=of q_2] {$q_4$};
			\node[state,accepting](q_5) [right=of q_3] {$q_5$};
			\node[state,accepting](q_6) [right=of q_4] {$q_6$};
			\path[->]
			(q_0) edge node [left] {$a, p<0, x:=p$} (q_1)
			(q_0) edge node {$a, p>0, x:=p$} (q_2)
			(q_1) edge [bend right] node [below] {$c, x+p=0, x:=p$} (q_3)
			(q_2) edge node {$a, x>0$} (q_4)
			(q_3) edge node {$a, x>0$} (q_5)
			(q_4) edge node {$b, \top$} (q_6)
			(q_5) edge node {$c, \top$} (q_6)
			(q_0) edge [bend right] node [below] {$a, p=0, x:=p$} (q_2)
			(q_2) edge [bend right] node [below] {$a, x=0$} (q_4)
			;
			\end{tikzpicture}
			
			\vspace{3em}
			\begin{tikzpicture}[shorten >=1pt,node distance=3cm,on grid,auto] 
			\node[state,initial,accepting] (q_0)   {$q_0$}; 
			\node[state,accepting] (q_1) [ below=of q_0] {$q_1$};
			\node[state,accepting] (q_2) [ right=of q_0] {$q_2$}; 
			\node[state,accepting](q_3) [right=of q_1] {$q_3$};
			\node[state,accepting](q_4) [right=of q_2] {$q_4$};
			\node[state,accepting](q_5) [right=of q_3] {$q_5$};
			\node[state,accepting](q_6) [right=of q_4] {$q_6$};
			\path[->]
			(q_0) edge node [left] {$a, p<0, x:=p$} (q_1)
			(q_0) edge node {$a, p>0, x:=p$} (q_2)
			(q_1) edge [bend right] node [below] {$c, x+p=0, x:=p$} (q_3)
			(q_2) edge node {$a, x>0$} (q_4)
			(q_3) edge node {$a, x>0$} (q_5)
			(q_4) edge node {$b, \top$} (q_6)
			(q_5) edge node {$c, \top$} (q_6)
			(q_0) edge node {$a, p=0, x:=p$} (q_3)
			(q_3) edge node [right] {$a, x=0$} (q_4)
			;
			\end{tikzpicture}
			\caption{There is no unique, coarsest location equivalence.}
			\label{fig: location equivalence not uniquely determined}
		\end{center}
	\end{figure}
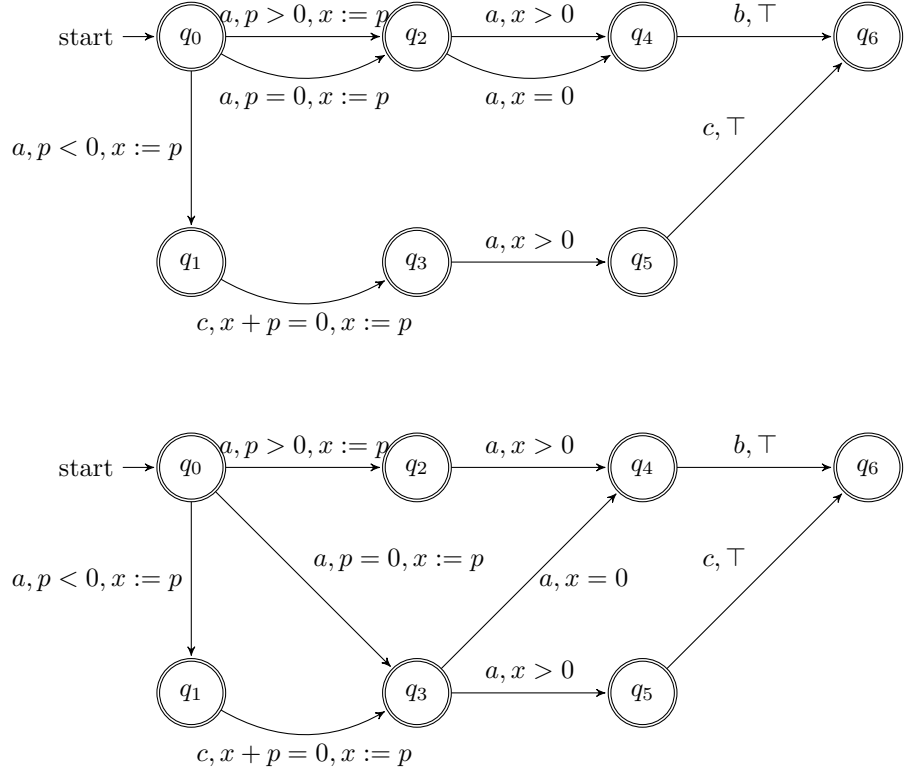
\end{example}

\begin{theorem}
	\label{theorem completeness}
	Suppose $L$ is a regular symbolic language over $\Sigma$.  
	Then there exists a register automaton $\A$ such that $L = L_s(\A)$.
\end{theorem}
\begin{proof}
	\iflong
	Let
		\else
	\emph{(Outline)} Let
	\fi
  $\equiv_l, \equiv_t, \equiv_r$ be relations satisfying the properties stated in Definition~\ref{def regular}.
	We define register automaton $\A =(\Sigma, Q, q_0, V, \Gamma)$ as follows:
	\begin{itemize}
		\item 
		$Q = \{[w]_l \mid  w \in L\}$.\\
		(Since $L$ is regular, $\equiv_l$ has finite index, and so $Q$ is finite, as required.)
		\item 
		$q_0 = [\epsilon]_l$.\\
		(Since $L$ is regular, it is feasible, and thus nonempty and prefix closed.  Therefore, $\epsilon \in L$.)
		\item 
		$V = \{ [(w,v)]_r  \mid w \in L \wedge v \in \V \wedge w \mbox{ stores } v \}$.\\
		(Since $L$ is regular, the equivalence induced by $\equiv_r$ has finite index, and so $V$ is finite, as required.  Note that registers are supposed to be elements of $\V$, and equivalence classes of $\equiv_r$ are not.  Thus, strictly speaking, we should associate a unique register of $\V$ to each equivalence class of $\equiv_r$, and define $V$ in terms of those registers.)
		\item 
		$\Gamma$ contains a transition $\langle q, \alpha, g, \varrho, q' \rangle$ for each equivalence class $[w \alpha G]_t$, where
		\begin{itemize}
			\item 
			$q = [w]_l$\\
			(Condition~\ref{cond t source} ensures that the definition of $q$ is independent from the choice of representative $w \alpha G$.)
			\item
			(Condition~\ref{cond t input} ensures that input symbol $\alpha$ is independent from the choice of representative $w \alpha G$.)
			\item
			$g \equiv G[\tau]$ where $\tau$ is a variable renaming that satisfies, for $v \in\Var(G)$,
			\begin{eqnarray*}
				\tau (v) & = & \left\{ \begin{array}{ll}
					[(w,v)]_r & \mbox{if } w \mbox{ stores } v\\
					p & \mbox{if } v=v_{m+1} \wedge m = \length(w)
				\end{array} \right.
			\end{eqnarray*}
			(By Condition~\ref{cond valuation defined for variables in guards}, $w \mbox{ stores } v$, for any  $v \in \Var(G)\setminus \{ v_{m+1} \}$, so $G[\tau]$ is well-defined.  Condition~\ref{cond t guard} ensures that the definition of $g$ is independent from the choice of representative $w \alpha G$.)
			Also note that, by Condition~\ref{cond r marker}, $\tau$ is injective.)
			\item
			$\varrho$ is defined for each equivalence class $[(w' \alpha G', v')]_r$ with
			$w'\alpha G'\equiv_t w \alpha G$ and $w' \alpha G' \mbox{ stores } v'$. Let $n = \length(w')$. Then
			\begin{eqnarray*}
				\varrho ([(w' \alpha G', v')]_r) & = &  \left\{ \begin{array}{ll}  [(w',v')]_r &\mbox{if } w' \mbox{ stores } v'\\
					p & \mbox{if } v' = v_{n+1}
				\end{array} \right.
			\end{eqnarray*}
			(By Condition~\ref{cond r backward propagation}, either $v' = v_{n+1}$ or $w'$ stores $v'$, so $\varrho ([(w' \alpha G', v')]_r)$ is well-defined. Also by Condition~\ref{cond r backward propagation}, the definition of $\varrho$ does not depend on the choice of representative $w'\alpha G'$.
			By Conditions~\ref{cond t parameter} and \ref{cond r forward propagation}, assignment $\varrho$ is injective.)
			\item
			$q' = [w \alpha G]_l$\\
			(Condition~\ref{cond t target} ensures that the definition of $q'$ is independent from the choice of representative $w \alpha G$.)
		\end{itemize}
		Since $L$ is regular, $\equiv_t$ has finite index and therefore $\Gamma$ is finite, as required.
		
		\iflong
		Note that in fact there exists a one-to-one correspondence between equivalence classes of $\equiv_t$ and the transitions in $\Gamma$.  Because suppose $w \alpha G \in L$ and $w' \alpha' G' \in L$ induce the same transition $\langle q, \alpha'', g, \varrho, q' \rangle$.
		Then $q = [w]_l = [w']_l$ and thus $w \equiv_l w'$.
		Also $\alpha = \alpha'' = \alpha'$ and thus $\alpha = \alpha'$.
		Moreover, $G[\tau] \equiv G'[\tau']$ (with $\tau'$ defined as expected). Now observe that
		$G[\tau] \equiv G[\sigma][\tau']$, for $\sigma = \matching(w,w')$.
		Thus we have $G[\sigma][\tau'] \equiv G'[\tau']$.  Since $\tau'$ is injective, this implies $G[\sigma] \equiv G'$.
		Now Lemma~\ref{lem determinism corollary} implies $w \alpha G  \equiv_ tw' \alpha' G'$.
		So each transition of $\Gamma$ corresponds to exactly one equivalence class of $\equiv_t$. 
		\fi
	\end{itemize}
	We claim that $\A$ is deterministic and prove this by contradiction.
	Suppose $\langle q, \alpha, g', \varrho', q' \rangle$ and $\langle q, \alpha, g'', \varrho'', q'' \rangle$ are two distinct $\alpha$-transitions in $\Gamma$ with $g' \wedge g''$ satisfiable.
	Then there exists a valuation $\xi$ such that $\xi \models g' \wedge g''$.
	Let the two transitions correspond to (distinct) equivalence classes $[w' \alpha G']_t$ and $[w'' \alpha G'']_t$, respectively. Then $g' = G'[\tau']$ and $g'' = G''[\tau'']$, with $\tau'$ and $\tau''$ defined as above.
	Now observe that $G'[\tau'] \equiv G'[\sigma][\tau'']$, for $\sigma = \matching(w',w'')$.  Using Lemma~\ref{lemma injective valuations}, we derive
	\[
	\xi \models g' \wedge g''  \Leftrightarrow
	\xi \models G'[\sigma][\tau''] \wedge G''[\tau''] \Leftrightarrow 
	\xi \models (G'[\sigma] \wedge G'')[\tau''] \Leftrightarrow 
	\xi \circ \tau'' \models G'[\sigma] \wedge G''.
	\]
	Thus $G'[\sigma] \wedge G''$ is satisfiable and we may apply Condition~\ref{cond determinism} to conclude
	$w' \alpha G' \equiv_t w'' \alpha G''$. Contradiction.
	
	So using the assumption that $L$ is regular, we established that $\A$ is a register automaton. Note that for this we essentially use that equivalences $\equiv_l$, $\equiv_t$ and $\equiv_r$ have finite index, as well as all the conditions, except Condition~\ref{cond right invariance}.
	\iflong
	
	It remains to prove $L = L_s(\A)$. 
	First, we show that $L \subseteq L_s(\A)$. 
	For this, suppose that $w = \alpha_1 G_1 \cdots\alpha_n G_n \in L$. We need to prove $w \in L_s(\A)$.  Consider the following sequence
	\[
	\delta ~=~ (q_0, \zeta_0) ~ \xrightarrow{\alpha_1, g_1, \varrho_1} ~ 
	(q_1, \zeta_1) ~
	\ldots ~ \xrightarrow{\alpha_n, g_n, \varrho_n} ~
	(q_n, \zeta_n),
	\]
	where $q_0 = [ w_0 ]_l$, $w_0 = \epsilon$, $\domain(\zeta_0) = \emptyset$ and, for $1 \leq i \leq n$,
	\begin{itemize}
		\item 
		$q_i = [ w_i ]_l$, where $w_i = \alpha_1 G_1 \cdots \alpha_i G_i$,
		\item
		$\langle q_{i-1}, \alpha_i, g_i, \varrho_i, q_i \rangle$ is the transition associated to $[w_i]_t$,
		\item
		$\zeta_i = \iota_i \circ \varrho_i$, where $\iota_i = \zeta_{i-1} \cup \{ (p, v_i) \}$.
	\end{itemize}
	Since $L$ is feasible, $G_1 \wedge \cdots \wedge G_n$ is satisfiable. Therefore, in order to prove that $\delta$ is a symbolic run of $\A$, it suffices to show,
	for $1 \leq i \leq n$,
	\begin{eqnarray*}
		G_i & \equiv & g_i [ \iota_i ].
	\end{eqnarray*} 
	Suppose $w_i$ stores $v$. Then, for $i>0$,
	\begin{eqnarray*}
		\varrho_i ([(w_i, v)]_r) & = &  \left\{ \begin{array}{ll}  [(w_{i-1},v)]_r &\mbox{if } w_{i-1} \mbox{ stores } v\\
			p & \mbox{if } v' = v_i
		\end{array} \right.
	\end{eqnarray*}
	By induction on $i$ we prove that $w_i \mbox{ stores } v \Rightarrow \zeta_i([w_i, v)]_r) = v$.
	\begin{itemize}
		\item 
		Base $i=0$. Trivial since $w_0$ does not store any $v$.
		\item
		Induction step. Assume $i>0$ and $w_i$ stores $v$.  We consider two cases:
		\begin{itemize}
			\item 
			$v = v_i$. Then $\zeta_i ( [(w_i, v)]_r ) = \iota_i \circ \varrho_i ( [(w_i, v_i)]_r) = \iota_i (p) = v_i =v$.
			\item
			$w_{i-1}$ stores $v$. Then
			\begin{eqnarray*}
				\zeta_i ( [(w_i, v)]_r ) & = &  \iota_i \circ \varrho_i ( [(w_i, v)]_r) = \iota_i ([(w_{i-1},v)]_r) \\
				& = & \zeta_{i-1} ([(w_{i-1},v)]_r) = v \mbox{ (by induction hypothesis)}.
			\end{eqnarray*}
		\end{itemize}
	\end{itemize}
	By definition $g_i \equiv G_i [\tau_i]$, where for $v \in \Var(G_i)$,
	\begin{eqnarray*}
		\tau_i (v) & = & \left\{ \begin{array}{ll}
			[(w_{i-1},v)]_r & \mbox{if } w_{i-1} \mbox{ stores } v\\
			p & \mbox{if } v = v_i
		\end{array} \right.
	\end{eqnarray*}
	This means we need to prove $G_i \equiv G_i [\tau_i] [ \iota_i ]$, that is, we must show, for $v \in \Var(G_i)$, that $\iota_i (\tau_i(v)) = v$. There are two cases:
	\begin{itemize}
		\item 
		If $v = v_i$ then $\iota_i (\tau_i(v)) = \iota_i (p) = v_i = v$.
		\item
		If $w_{i-1}$ stores $v$  then $\iota_i (\tau_i(v)) = \iota_i ( [(w_i, v)]_r) =  \zeta_i ( [(w_i, v)]_r) = v$.
	\end{itemize}
	We conclude that $\delta$ is a symbolic run with $\strace(\beta) = w$.
	Since $w \in L$, $q_n = [w]_l \in F$, so symbolic run $\beta$ is accepting, and thus $w \in L_s(\A)$, as required.
	
Next we need to show that $L_s(\A)\subseteq L$.
	For this, suppose $w = \alpha_1 G_1 \cdots\alpha_n G_n \in L_s(\A)$. We need to prove $w \in L$.
	Let
	\[
	\delta ~=~ (q_0, \zeta_0) ~ \xrightarrow{\alpha_1, g_1, \varrho_1} ~ 
	(q_1, \zeta_1) ~
	\ldots ~ \xrightarrow{\alpha_n, g_n, \varrho_n} ~
	(q_n, \zeta_n),
	\]
	be a symbolic run of $\A$, as in Definition~\ref{def symbolic semantics}, with $\strace(\delta) = w$. 
	For $0 < i \leq n$, suppose transition $\langle q_{i-1}, \alpha_i, g_i, \varrho_i, q_i \rangle$ corresponds to
	equivalence class $[ u_{i-1} \alpha_i G'_i]_t$.
	For $0 \leq i \leq n$, let $w_i = \alpha_1 G_1 \cdots \alpha_i G_i$. 

	We prove by induction that $q_i = [w_i]_l$ and $w_i$ stores $v \Rightarrow \zeta_i ([(w_i,v)]_r) = v$.
	\begin{itemize}
		\item 
		Base $i=0$. Trivial, since $q_0 = [\epsilon]_l  = [w_0]_l$ and $\domain(\zeta_0) = \emptyset$ by definition.
		\item
		Induction step. Assume $i>0$. Since transition $q_{i-1} \xrightarrow{\alpha_i, g_i, \varrho_i} q_i$ corresponds to equivalence class $[u_{i-1} \alpha_i G'_i]_t$, $q_{i-1} = [u_{i-1}]_l$.
		%Then $G_i = g_i[\iota_i]$, where $\iota_i = $
		Therefore, by induction hypothesis, $u_{i-1} \equiv_l w_{i-1}$.
		By Definition~\ref{def symbolic semantics}, $G_i \equiv g_i [\iota_i]$ and by definition of $\A$, $g_i \equiv G'_i[\tau]$, where for each $v \in\Var(G'_i)$,
		\begin{eqnarray*}
			\tau (v) & = & \left\{ \begin{array}{ll}
				[(u_{i-1},v)]_r & \mbox{if } u_{i-1} \mbox{ stores } v\\
				p & \mbox{if } v=v_{m+1}
			\end{array} \right.
		\end{eqnarray*}
		where $m = \length(u_{i-1})$. Thus $G_i \equiv  G'_i[\iota_i \circ \tau]$.
		Let $\sigma = \matching(u_{i-1}, w_{i-1})$. Then, for each $v \in\Var(G'_i)$, $\iota_i \circ \tau(v) = \sigma(v)$:
		\begin{itemize}
			\item 
			If $v=v_{m+1}$ then
			$\iota_i \circ \tau (v) = \iota_i (p) = v_i = \sigma(v)$.
			\item 
			If $v \neq v_{m+1}$ then, by Condition~\ref{cond valuation defined for variables in guards}, there exists a $v'$ such that $(u_{i-1},v) \equiv_r (w_{i-1}, v')$. Then, again by induction hypothesis,
			\[
			\iota_i \circ \tau(v) = \iota_i ([(u_{i-1},v)]_r) = \zeta_{i-1} ([(u_{i-1},v)]_r) = \zeta_{i-1}([(w_{i-1}, v')]_r) = v' = \sigma(v).
			\]
		\end{itemize}
		Therefore $G_i \equiv G'_i[\sigma]$.
		Since $\delta$ is a symbolic run, $\guard(w_{i-1}) \wedge G_i$ is satisfiable.
		Now we may use Condition~\ref{cond right invariance} to conclude $w_i = w_{i-1} \alpha_i G_i \in L$.
		Then, by Lemma~\ref{lem determinism corollary}, $u_{i-1} \alpha_i G'_i \equiv_t w_i$, and thus, by Condition~\ref{cond t target}, $u_{i-1} \alpha_i G'_i \equiv_l w_i$. From this, we conclude $q_i = [w_i]_l$.
		
			Suppose $w_i$ stores $v$. Since $u_{i-1} \alpha_i G'_i \equiv_t w_i$,
		\begin{eqnarray*}
			\varrho_i ([(w_i, v)]_r) & = &  \left\{ \begin{array}{ll}  [(w_{i-1},v)]_r &\mbox{if } w_{i-1} \mbox{ stores } v\\
				p & \mbox{if } v' = v_i
			\end{array} \right.
		\end{eqnarray*}
		 Assume $w_i$ stores $v$.  We consider two cases:
			\begin{itemize}
				\item 
				$v = v_i$. Then $\zeta_i ( [(w_i, v)]_r ) = \iota_i \circ \varrho_i ( [(w_i, v_i)]_r) = \iota_i (p) = v_i =v$.
				\item
				$w_{i-1}$ stores $v$. Then, using the induction hypothesis,
				\begin{eqnarray*}
					\zeta_i ( [(w_i, v)]_r ) & = &  \iota_i \circ \varrho_i ( [(w_i, v)]_r) = \iota_i ([(w_{i-1},v)]_r) \\
					& = & \zeta_{i-1} ([(w_{i-1},v)]_r) = v.
				\end{eqnarray*}
			\end{itemize}
		
	\end{itemize}
Thus in particular $q_n = [w_n]_l = [w]_l$. This implies $w \in L$, as required.

As a final note, we observe that $\A$ is well-formed.
Because suppose $\delta$ is a symbolic run that ends with $(q, \zeta)$ and suppose
$q \xrightarrow{\alpha, g \varrho} q'$.
Let transition $q \xrightarrow{\alpha, g \varrho} q'$ correspond to equivalence class $[w \alpha G]_t$.
Suppose $x \in \Var(g)$. Then, by construction of $\A$, there is a variable $v \in \Var(G)$ such that either
$x = [(w,v)]_r$ and $w$ stores $v$, or $x = p$ and $v = v_{m+1}$, where $m = \length(w)$.
Let $w' = \strace(\delta)$.
By the above inductive proof, $q = [w']_l$ and $w'$ stores $v' \Rightarrow [(w',v')]_r \in \domain(\zeta)$.
Then $w \equiv_l w'$ and by Condition~\ref{cond valuation defined for variables in guards}, either $v = v_{m+1}$ or
there exists a $v'$ such that $(w,v) \equiv_r (w',v')$.
This means that either $x=p$ or $x \in\domain(\zeta)$.
Hence we may conclude that $\Var(g) \subseteq \domain(\zeta) \cup \{ p \}$ and thus $\A$ is well-formed by Corollary~\ref{variables guards always defined V2}. \qed
\else
We claim $L = L_s(\A)$.
\fi
\end{proof}

\section{Concluding Remarks}
We have shown that register automata can be defined in a natural way \emph{directly} from a regular symbolic language, with locations materializing as equivalence classes of a relation $\equiv_l$, transitions as equivalence classes of a relation $\equiv_t$, and registers as equivalences classes of a relation $\equiv_r$.

It is instructive to compare our definition of regularity for symbolic languages with Nerode's original definition for non-symbolic languages. Nerode defined his equivalence for all words $u, v \in \Sigma^{\ast}$ (not just those in $L$!) as follows:
\begin{eqnarray*}
%\label{NerodeEquivalence}
	u \equiv_l v & \Leftrightarrow & (\forall w \in \Sigma^{\ast} : u w \in L \Leftrightarrow v w \in L).
\end{eqnarray*}
For any language $L \subseteq \Sigma^{\ast}$, the equivalence relation $\equiv_l$ is uniquely determined and can be used (assuming it has finite index) to define a unique minimal finite automaton that accepts $L$.
As shown by Example~\ref{example location equivalence not unique}, the equivalence $\equiv_l$ and its corresponding register automaton are not uniquely defined in a setting of symbolic languages.  For such a setting, it makes sense to consider a symbolic variant of what Kozen \cite{Kozen97} calls \emph{Myhill-Nerode relations}. These are relations that satisfy the following three conditions,
for $u, v \in \Sigma^{\ast}$ and $\alpha \in \Sigma$,
\begin{eqnarray}
u \equiv_l v & \Rightarrow & (u  \in L \Leftrightarrow v  \in L) \label{accepting}\\
u \equiv_l v & \Rightarrow & u \alpha \equiv_l v \alpha \label{right invariance}\\
\equiv_l & \mbox{has} & \mbox{finite index} 
\end{eqnarray}
Note that Conditions ~\ref{accepting} and \ref{right invariance} are consequences of Nerode's definition. Condition \ref{right invariance} is the well-known right invariance property, which is sound for non-symbolic languages, since finite automata are completely specified and every state has an outgoing $\alpha$-transition for every $\alpha$. A corresponding condition
\begin{eqnarray*}
u \equiv_l v & \Rightarrow & u \alpha G \equiv_l v \alpha G 
\end{eqnarray*}
for symbolic languages would not be sound, however, since locations in a register automaton do not have outgoing transitions for every possible symbol $\alpha$ and every possible guard $G$.
We see basically two routes to fix this problem. The first route is to turn $\equiv_l$ into a partial equivalence relation that is only defined for symbolic words that correspond to runs of the register automaton. Right invariance can then be stated as
\begin{eqnarray}
	\label{right invariance PER}
	&& w \equiv_l w'  \wedge w \alpha G \equiv_l w \alpha G \wedge \sigma=\matching(w, w')  \wedge  w' \alpha G[\sigma] \equiv_l w' \alpha  G[\sigma] \nonumber\\
	&& \quad\quad  \Rightarrow  w \alpha G \equiv_l w' \alpha G[\sigma].
\end{eqnarray}
The second route is to define $\equiv_l$ as an equivalence on $L$ and restrict attention to prefix closed symbolic languages.  This allows us to drop Condition~\ref{accepting} and leads to the version of right invariance that we stated as Condition~\ref{cond right invariance}.
Since prefix closure is a natural restriction that holds for all the application scenarios we can think of, and since equivalences are conceptually simpler than PERs, we decided to explore the second route in this article. However, we conjecture that the restriction to prefix closedness is not essential, and Myhill-Nerode characterization for symbolic trace languages without this restriction can be obtained using Condition~\ref{right invariance PER}.

An obvious research challenge is to develop a learning algorithm for symbolic languages based on our Myhill-Nerode theorem. Since for symbolic languages there is no unique, coarsest Nerode congruence that can be approximated, as in Angluin's algorithm \cite{Ang87}, this is a nontrivial task. We hope that for register automata with a small number of registers, an active algorithm can be obtained by encoding symbolic traces and register automata as logical formulas, and using SMT solvers to generate hypothesis models, as in \cite{SmetsersFV18}.

As soon as a learning algorithm for symbolic traces has been implemented, it will be possible to connect the implementation with the setup of \cite{TaintingRALib}, which extracts symbolic traces from Python programs using an existing tainting library for Python. We can then compare its performance with the grey-box version of the RALib tool \cite{TaintingRALib} on a number of benchmarks, which include data structures from Python's standard library.
An area where learning algorithms for symbolic traces potentially can have major impact is the inference of behavior interfaces of legacy control software. As pointed out in \cite{JasperMMSHSSHSK19},
such interfaces allow components to be developed, analyzed, deployed and maintained in isolation. This is achieved using enabling techniques, among which are model checking (to prove interface compliance), observers (to check interface compliance), armoring (to separate error handling from component logic) and test generation (to increase test coverage).
Recently, automata learning has been applied to 218 control software components of ASML’s TWINSCAN lithography machines \cite{YangASLHCS19}. Using black-box learning algorithms in combination with information from log files, 118 components could be learned in an hour or less. The techniques failed to successfully infer the interface protocols of the remaining 100 components. It would be interesting to explore whether grey-box learning algorithm can help to learn models for these and even more complex control software components.

\paragraph{Acknowledgements}
We thank Joshua Moerman, Thorsten Wi{\ss}mann and the anonymous reviewers for valuable feedback on earlier versions of this article.

\section*{References}
\bibliographystyle{elsarticle-harv}
\bibliography{abbreviations,dbase}

\end{document}